\documentclass[11pt]{article}

\usepackage[top=1in, bottom=1in, left=1in, right=1in]{geometry}

\usepackage[utf8]{inputenc} 
\usepackage[T1]{fontenc}
\usepackage{url}
\usepackage{ifthen}
\usepackage{cite}
\usepackage[cmex10]{amsmath} 

\usepackage{cite}
\usepackage[cmex10]{amsmath}
\usepackage{amssymb}
\usepackage{stfloats}

\usepackage{amsthm} 
\usepackage{wrapfig}
\usepackage{amssymb}
\usepackage{latexsym}
\usepackage{bm, bbm} 
\usepackage{epsfig}
\usepackage{graphicx}
\usepackage{epsfig}
\usepackage{multicol}
\usepackage{multirow, array}
\usepackage{psfrag}
\usepackage{cite}
\usepackage{subfig}
\usepackage{float,subfloat}
\usepackage{color}
\usepackage{mathrsfs}

\usepackage{tkz-graph}
\usepackage{tikz}
\usetikzlibrary{shapes.geometric}
\usepackage[noend]{algpseudocode} 
	
	\newcommand{\Input}{\State \textbf{Input: }}
\usepackage{algorithm} 
\usepackage{algorithmicx} 
\usepackage{tabularx} 

\usepackage{mathtools}

\usepackage{authblk}

\usepackage{orcidlink}

\newtheorem{lemma}{Lemma}[section]
\newtheorem{definition}[lemma]{Definition}
\newtheorem{theorem}[lemma]{Theorem}
\newtheorem{prop}[lemma]{Proposition}
\newtheorem{corollary}[lemma]{Corollary}

\theoremstyle{plain}

\theoremstyle{plain}

\newtheorem{PreRemark}[lemma]{{\textbf{Remark}}}
  \newenvironment{remark}%
    {\begin{PreRemark}\upshape}{\hfill$\square$\end{PreRemark}}

\newtheorem{PreExample}[lemma]{{\textbf{Example}}}
  \newenvironment{example}%
    {\begin{PreExample}\upshape}{\hfill$\square$\end{PreExample}}

\long\def\symbolfootnote[#1]#2{\begingroup%
\def\thefootnote{\fnsymbol{footnote}}\footnote[#1]{#2}\endgroup} 

\makeatletter
\newcommand{\multiline}[1]{%
  \begin{tabularx}{\dimexpr\linewidth-\ALG@thistlm}[t]{@{}X@{}}
    #1
  \end{tabularx}
}
\makeatother

\MakeRobust{\Call}

\makeatletter
\renewcommand*\env@matrix[1][*\c@MaxMatrixCols c]{%
  \hskip -\arraycolsep
  \let\@ifnextchar\new@ifnextchar
  \array{#1}}
\makeatother

\allowdisplaybreaks[4]
\usepackage{amssymb,amsmath,url,etoolbox,multirow,multicol,enumerate}
\usepackage{graphicx,epstopdf,epsfig}
\usepackage{tikz-cd} 
\usepackage{inputenc}
\usepackage{amsmath,  amssymb,amsbsy,amsfonts,amsthm,latexsym,amsopn,amstext, amsxtra,euscript,amscd}
\usepackage{physics}

\usepackage{makecell}
\usepackage{mathtools}
\usepackage{amsthm}
\usepackage{amsmath}
\usepackage{amssymb}
\usepackage{tikz}

\usetikzlibrary{arrows}

\newtheorem{theo}{Theorem}
\newtheorem{cor}[theo]{Corollary}
\newtheorem{lem}[theo]{Lemma}
\theoremstyle{definition}

\newtheorem{exa}[theo]{Example}

\newtheorem{rem}[theo]{Remark}

\newtheorem{construction}[theo]{Construction}

\numberwithin{theo}{section}

\newcommand{\rmv}[1]{}

\newcommand{\F}{{\mathbb F}}

\newcommand{\Z}{{\mathbb Z}}

\newcommand{\C}{{\mathbb C}}

\newcommand{\cC}{{\mathcal C}}
\newcommand{\cD}{{\mathcal D}}

\newcommand{\cN}{{\mathcal N}}

\newcommand{\cP}{{\mathcal P}}
\newcommand{\cQ}{{\mathcal Q}}

\newcommand{\cS}{{\mathcal S}}

\newcommand{\bzero}{\mathbf{0}}
\newcommand{\ba}{{\mathbf a}}

\newcommand{\bc}{{\mathbf c}}

\newcommand{\be}{{\mathbf e}}

\newcommand{\bh}{{\mathbf h}}
\newcommand{\bi}{{\mathbf i}}
\newcommand{\bj}{{\mathbf j}}
\newcommand{\bk}{{\mathbf k}}
\newcommand{\bl}{{\mathbf l}}

\newcommand{\bs}{{\mathbf s}}

\newcommand{\bu}{{\mathbf u}}
\newcommand{\bv}{{\mathbf v}}
\newcommand{\bw}{{\mathbf w}}
\newcommand{\bx}{{\mathbf x}}
\newcommand{\by}{{\mathbf y}}
\newcommand{\bz}{{\mathbf z}}

\renewcommand{\o}{\omega}

\AtBeginDocument{}

\def\s{\sigma}

\newcommand{\balpha}{{\boldsymbol \alpha}}

\renewcommand{\sup}{{\sf supp}}
\renewcommand{\dim}{{\sf dim}}
\renewcommand{\ker}{{\sf ker}}

\newcommand{\rk}{{\sf rank}}

\newcommand{\rs}{{\sf rs}}

\newcommand{\wt}{{\sf wt}}

\def\be{\begin{equation}}
\def\ee{\end{equation}}
\def\matrix#1{\begin{bmatrix}#1\end{bmatrix}} 

\newcounter{alp}
\newcounter{ara}
\newcounter{rom}

\newenvironment{arabiclist}{\begin{list}{(\arabic{ara})\hfill}{\usecounter{ara}
     \topsep0.4ex \labelwidth.6cm \leftmargin.6cm \labelsep0cm
     \rightmargin0cm \parsep0ex \itemsep0ex}}{\end{list}}

\newcommand*{\boldone}{\text{\usefont{U}{bbold}{m}{n}1}}
\newcommand*{\boldzero}{\text{\usefont{U}{bbold}{m}{n}0}}

\newcommand{\sfX}{{\sf X}}

\newcommand{\sfZ}{{\sf Z}}



\begin{document}
\title{Combinatorial Analysis of Dyadic and Quasi-Dyadic Codes}

\author[1]{Anthony G\'omez-Fonseca\orcidlink{0000-0002-3569-9304}}
\author[2]{Gretchen L. Matthews\orcidlink{0000-0002-8977-8171}\thanks{G. L. Matthews is partially supported by NSF DMS-2502705 and the Commonwealth Cyber Initiative.}}
\author[2]{Kirsten D. Morris\orcidlink{0009-0003-9554-2290}}
\author[1]{Tefjol Pllaha\orcidlink{0000-0001-6280-2648r}}
\affil[1]{Department of Mathematics \& Statistics, University of South Florida, Tampa, FL USA}
\affil[ ]{\{agomezfonseca, tpllaha\}@usf.edu}
\affil[2]{Department of Mathematics, Virginia Tech, Blacksburg, VA USA}
\affil[ ]{\{gmatthews, kdmorris\}@vt.edu}


%
%
%
%
%
%
%
%
%
%

\date{}

\maketitle

\begin{abstract}
Quantum low-density parity-check (QLDPC) codes offer a promising route to scalable fault-tolerant quantum computation, but their performance under iterative decoding is strongly influenced by short-cycle structure and other harmful subgraphs in the associated Tanner graphs. This paper develops an algebraic framework for constructing and analyzing (Q)LDPC codes from dyadic and quasi-dyadic matrices—translation-invariant $2^\ell \times 2^\ell$ binary matrices specified compactly by a signature row and forming a commutative ring with recursive block structure. Leveraging this structure, we relate cycles in lifted Tanner graphs to tailless backtrackless closed walks in the protograph and derive efficient, implementable methods to enumerate and control short cycles (notably $4$-, $6$-, and $8$-cycles). We introduce dyadic-aware PEG-style construction algorithms that use forbidden sets of shifts to maximize attainable girth when possible and otherwise minimize the multiplicity of the shortest cycles at the target girth. Motivated by error-floor phenomena, we further characterize and explicitly enumerate absorbing sets in key dyadic layout boundary cases, identifying configurations that induce abundant $(a,0)$-absorbing sets. Finally, we propose CSS-oriented dyadic constructions that satisfy commutation constraints by design and demonstrate via belief-propagation simulations that reducing short-cycle multiplicity can yield substantial decoding gains even when girth cannot be increased.
\end{abstract}


%


\section{Introduction}
Quantum error correction is a central prerequisite for scalable quantum computing, and the current momentum around \emph{quantum low-density parity-check} (QLDPC) codes reflects a broader push toward code families that combine strong protection with implementable, sparse stabilizer measurements. In the CSS framework, this reduces to constructing sparse parity-check matrices $H_X$ and $H_Z$ that satisfy the commutation constraint $H_{\sf Z}H_{\sf X}^{\top}=0$ while still producing Tanner graphs with favorable combinatorial structure---especially large girth and controlled populations of short cycles. Short cycles are well known to degrade iterative decoding and to raise the error floor; moreover, in quantum settings, they can exacerbate syndrome-based iterative decoder failures through small problematic substructures \cite{dolecek2009analysis,Morris2026AMC}. 

In this work we study---and exploit---the algebraic structure of \emph{dyadic} and \emph{quasi-dyadic} matrices to build and analyze families of (Q)LDPC codes. A dyadic matrix $D\in\mathbb{F}_2^{N\times N}$ (with $N=2^\ell$) is determined by a \emph{signature} row $\sigma(D)$ via the translation-invariant rule $D_{\bx,\by}=D_{{\bf 0},\bx+\by}$, yielding a commutative ring with strong recursive block structure. This dyadic framework provides two key advantages for code design: (i) the lifting factor is naturally a power of two, and (ii) products and sums of dyadic permutation matrices encode walk composition in lifted Tanner graphs through simple ``permutation shifts.'' These features make dyadics an appealing substrate for protograph-based constructions and for CSS constraints, while enabling a level of cycle accounting that is typically costly for unstructured LDPC ensembles \cite{Gulamhusein,mpk24}.

Our first contribution is a \emph{cycle and girth analysis toolkit} specialized to dyadic and quasi-dyadic lifts. We connect cycles in the lifted Tanner graph to tailless backtrackless closed (TBC) walks in the protograph and show how powers of the (block) adjacency matrix enumerate protograph walks whose permutation-shift labels determine whether they lift to genuine cycles. Building on this correspondence, we derive efficient counting strategies for short cycles---starting with explicit, implementable conditions and formulas for enumerating $4$-cycles in both single-dyadic and protograph-based quasi-dyadic constructions, and extending the analysis to $6$- and $8$-cycles. These counts are not merely descriptive: they become design constraints and objective functions for construction and optimization \cite{gfsm23b, gfsm23c}.

Second, we introduce \emph{dyadic-aware PEG-style construction methods} that leverage ``forbidden sets'' of permutation shifts to avoid (or delay) the formation of short cycles during matrix assembly. Because choosing a single $1$ in a dyadic signature row induces an entire dyadic permutation block, dyadic lifting changes the effective granularity of edge placement---and this, in turn, allows a substantial reduction in construction complexity compared to applying PEG directly to the full lifted Tanner graph. We present practical algorithms that (i) maximize attainable girth when possible and (ii) further \emph{minimize the number of short cycles} at the target girth when avoidance is no longer feasible \cite{hea05,gfsm23b}.

Third, motivated by the role of harmful substructures in iterative decoding, we analyze \emph{absorbing sets} in Tanner graphs arising from dyadic layouts. We characterize boundary cases (e.g., $m\times 1$ and $1\times n$ arrays of dyadic permutations) and identify configurations that yield disconnected unions of complete bipartite components, for which absorbing sets can be enumerated explicitly. This clarifies when dyadic structure may inadvertently introduce an abundance of $(a,0)$-absorbing sets and when certain layouts can be ruled out due to unavoidable problematic combinatorics \cite{dolecek2009analysis,mcmillon2023extremal}.

Finally, we propose and evaluate CSS-oriented \emph{code constructions from dyadic permutations} that naturally satisfy commutation constraints through structured pairing and permutation choices, and we compare them to recent quasi-dyadic CSS constructions whose Tanner graphs inevitably have girth four. Using belief-propagation simulations as a proof of concept, we demonstrate that even when girth cannot be increased, reducing the multiplicity of the shortest cycles can yield substantial decoding gains; when larger girth is achievable, combining girth control with short-cycle minimization improves performance further \cite{BBS26, kasai}.

This paper is organized as follows. This section concludes with a review of notation to be used throughout the paper. 
Section \ref{sec:preliminaries} reviews dyadic matrices, protograph-based LDPC codes, and CSS stabilizer prerequisites. Section \ref{sec:code_construction} details the code constructions studied in this paper.
Section \ref{sec:girth_analysis} develops the dyadic/protograph TBC-walk framework and presents explicit cycle-counting results and algorithms. Section \ref{sec:absorbing_set_analysis} studies absorbing sets in dyadic-induced Tanner graphs. Section \ref{sec:simulation_results} presents some simulation comparisons. Section \ref{sec:conclusions} concludes with open directions, and Appendix \ref{appendix:girth_analysis} contains additional short-cycle enumerations beyond the main text.
\\

\textbf{Notation. } All vectors will be row vectors and will be denoted with bold letters. Matrices will be denoted with italic capital letters. The all-zero and all-one vectors in $\F_2^n$ will be denoted ${\bf 0}_n$ and ${\bf 1}_n$, respectively. 
The zero matrix and the all-one matrix will be denoted $\boldzero_n$ and $\boldone_n$, respectively. The subscript will be dropped when the context is clear. 
The row space and the (right) kernel of a matrix $A$ will be denoted $\rs(A)$ and $\ker(A)$,  respectively.
The space of binary $m\times n$ matrices is denoted by $\F_2^{m\times n}$. Given $A \in \F_2^{m \times n}$, the entry of $A$ in Column $i$ and Row $j$ is denoted by $A_{i,j}$. The transpose of $A$ is denoted by $A^\top$.
We will use the notation $[n:m]:=\{n,n+1,\ldots,m\}$ and $[n]:=[1:n]$. \rmv{ and $[n]_0:=[0:n-1]$.}

A binary linear $[n,k]$ code $\cC$ is a $k$-dimensional subspace of $\F_2^n$. We will use the term code to mean binary linear code. The parameter $n$ is called the length of the code.
The support of a vector $\bx\in\F_2^n$ is $\sup(\bx):= \{i\in [n]\mid x_i \neq 0 \}$.
The Hamming weight of a vector $\bx$ is $\wt(\bx):=|\sup(\bx)|$, and the Hamming distance of  $\cC$ is $d(\cC) = {\sf min}\{\wt(\bx)\mid \bx\in\cC, \bx\neq {\bf 0}_n\}$.
If the code has Hamming distance $d$, then it is called an $[n,k,d]$ code.
For a code $\cC$, a matrix $M$ is called a generator matrix if $\rs(M) = \cC$ and a parity-check matrix if $\ker(M) = \cC$. Note here that we require neither generator matrices nor parity-check matrices to be full rank. The dual of  $\cC$ is defined as $\cC^\perp:=\{\bx\in \F_2^n \mid \bx\bc^\top = 0, \text{ for all }\bc\in \cC\}$.
We have that $\dim(\cC^\perp) = n-k$. By the definition, it follows that a generator matrix of a code is a parity-check matrix of its dual and vice-versa, that is, $\cC = \rs(G) = \ker(H)$ if and only if $\cC^\perp = \rs(H) = \ker(G)$. 
A code $\cC$ is called self-orthogonal (resp., self-dual) if $\cC \subseteq \cC^\perp$ (resp., $\cC = \cC^\perp$). Note that for self-dual codes we necessarily have $k = n/2$ and self-dual codes of odd length cannot exist. Analogously, a code $\cC$ is called dual-containing if $\cC^\perp \subseteq \cC$. 
Since $(\cC^\perp)^\perp = \cC$, we have that $\cC$ is self-orthogonal if and only if $\cC^\perp$ is dual-containing.
The focus of this paper is \textbf{low-density parity-check codes} (LDPC), that is, codes $\cC = \ker(H)$ where $H$ is a sparse matrix.

\section{Preliminaries}
\label{sec:preliminaries}

\subsection{Dyadic Matrices}
Let $\ell$ be a natural number, and let $N = 2^\ell$. We will consider $N \times N$ binary matrices whose rows and columns will be indexed using elements of $\F_2^\ell$ via the bijection 
\begin{equation}\label{e-bijection}
\bx = (x_1,x_2,\ldots,x_{\ell})\in \F_2^\ell \leftrightarrow x = 1+\sum_{i=1}^\ell x_i 2^{i-1} \in [N].  
\end{equation}
Given $M \in \F_2^{N \times N}$,
under this bijection, the elements of $\F_2^{\ell-1}\times \{0\}$ index the ``left'' half of columns and  ``top'' half of rows, meaning columns $[N/2]$ and rows $[N/2]$ of $M$,  whereas the elements of $\F_2^{\ell-1}\times \{1\}$ index the ``right'' half of columns and the ``bottom'' half of rows, meaning columns $[N/2+1:N]$ and Rows $[N/2+1:N]$ of $M$. 
\begin{definition}
\label{definition_dyadic_matrix}
A matrix $D\in \F_2^{N\times N}$ is called a {\bf dyadic matrix} if its entries satisfy $D_{\bx,\by} = D_{{\bf 0},\bx+\by}$ for all $\bx, \by \in \F_2^\ell$.
The {\bf signature} of a dyadic matrix $D$, denoted $\s(D)$, is the first row $D_{{\bf 0},\bullet}$ of $D$. For $\bv \in \F_2^N$, we will denote by $D_\bv$ the corresponding dyadic matrix with signature $\bv$, that is, $\s(D_\bv) = \bv$.
A \textbf{quasi-dyadic matrix} is a block matrix where its blocks are dyadic matrices.
\end{definition}
It is well-known~\cite{Gulamhusein} that the collection of $N\times N$ dyadic matrices is isomorphic to the polynomial ring $\F_2[x_1,\ldots,x_\ell]/\langle x_1^2+1,\ldots,x_\ell^2+1\rangle$. This in turn immediately implies that the set of $N \times N$ dyadic matrices forms a commutative ring, which we will denote $\cD_\ell$. 
Beyond commutativity, we list some additional elementary properties below; see~\cite{mpk24} and the references therein.
\begin{prop}\label{P-properties}
For $N=2^\ell$, the following hold.
\begin{arabiclist}
\item The set of $N \times N$ dyadic matrices  is 
\[
\cD_\ell = \left\{\matrix{A&B\\B&A} \middle| A,B\in \cD_{\ell-1} \right\}
\]where $\cD_0=\left\{ \left[ 0 \right],  \left[ 1 \right] \right\}$, and $\mid \cD_\ell \mid = 2^N$.
\item Dyadic matrices are symmetric.
\item The zero matrix and the identity matrix in $\F_2^{N\times N}$ are dyadic.
\item Let $D\in \cD_\ell$. Then, $D^2 = \boldzero_N$ if $\wt(\s(D))$ is even, and $D^2 = I_N$ if $\wt(\s(D))$ is odd. 
\item For any $\bu,\bv\in\F_2^N$, $D_{\bu} + D_\bv = D_{\bu+\bv}$. 
\end{arabiclist}
\end{prop}

Notice that for any $\bv \in \F_2^N$,
\begin{equation}
D_\bv = \matrix{D_\bu & D_\bw\\D_\bw & D_\bu}.
\end{equation}
where $\bv = (\bu,\bw)$ with $\bu,\bw\in \F_2^{N/2}$.

\begin{remark}\label{R-sup}
Let $D \in \cD_\ell$ be a dyadic matrix. Then, by definition, $\sup(\s(D)) \subseteq [N]$.
Under the bijection~\eqref{e-bijection}, we can view this support as a subset of $\F_2^\ell$.
Along this line, for $\bx\in \F_2^\ell$, we will denote $\tilde{\bx} \in \F_2^N$ the indicator vector with 1 in position $\bx$ and $0$ else. 
Thus, $\{\tilde{\bx}\mid \bx\in \F_2^\ell\}$ is the standard basis of $\F_2^N$.
\end{remark}
\begin{definition}
A dyadic matrix $D$ is called a {\bf dyadic permutation matrix} if $\wt(\s(D)) = 1$. 
\end{definition}

Notice that for all $\bx \in \F_2^\ell$, $D_{\tilde{\bx}}$ is a dyadic permutation matrix. For emphasis, we write $P_{\bx}:=D_{\tilde{\bx}}$. 
Alternatively, we will also denote the same dyadic permutation as $P_{x}$ where $x\in [N]$ is the integer that corresponds to $\bx$ via~\eqref{e-bijection}.
Even though the set $\cD_\ell$ is closed under multiplication, the signature of a product of dyadic matrices is not in general easily expressible in terms of the signatures of its factors. The next result shows that a nice expression is possible for dyadic permutation matrices.

\begin{prop}\label{P-permutation}
Let $\bx,\by\in \F_2^\ell$. Then $\sup(\s(D\cdot P_{\bx})) = \bx + \sup(\s(D))$ for all $D \in \cD_\ell$.
As a consequence, $P_{\bx} P_{\by} = P_{\bx +\by}$.
\end{prop}

A code $\mathcal C$ is said to be dyadic (resp., quasi-dyadic) if it has a parity-check matrix that is dyadic (resp., quasi-dyadic). Such codes are reproducible, as defined in \cite{SPB22}, and have compact representations. Indeed, it has a parity-check matrix that depends only on a single vector, its signature. We note that dyadic and quasi-dyadic codes may also have parity-check matrices that do not reflect this structure, since any matrix $A$ with $\ker(A)=\mathcal C$ is a parity-check matrix for $\mathcal C$.     


\subsection{LDPC Codes}

Let $\mathcal{C}$ be a quasi-dyadic LDPC code, either $(d_v,d_c)$-regular or irregular, with blocklength $2^\ell\cdot n_v$ based on the $n_c\times n_v$ \textbf{protograph} \cite{tho03} described by the matrix $B={(b_{ij})}_{n_c\times n_v}$, where $b_{ij}$ is a nonnegative integer for $i\in[n_c]$ and $j\in[n_v]$, and where $[l]=\{0,1,\dots,l-1\}$. Then $\mathcal{C}$ can be described by a (scalar) parity-check matrix $H={(H_{ij})}_{n_c\times n_v}$, where each $H_{ij}$, for $i\in[n_c]$ and $j\in[n_v]$, is a summation of $b_{ij}$ (non-overlapping) $2^\ell\times 2^\ell$ dyadic permutation matrices if $b_{ij}$ is nonzero, and the $2^\ell\times 2^\ell$ all-zero matrix if $b_{ij}=0$. Graphically, this operation is equivalent to taking a $2^\ell$-fold graph cover, or \textbf{lifting}, of the protograph. Here, $N=2^\ell$ is called the lifting factor. 

From the parity-check matrix $H$, we construct a bipartite graph $G=(V,E)$, called a \textbf{Tanner graph} \cite{tan81}, by considering $H$ as its biadjacency matrix.
This bipartite graph represents the quasi-dyadic LDPC code $\mathcal{C}$ obtained from $H$. The set $V$ is the set of vertices (or nodes) and $E$ is the set of edges. Label the vertices of $G$ by $v_a$, for $a=0,1,2,\dots,|V|-1$, and the edges by $e_b$, for $b=0,1,2,\dots,|E|-1$. Each edge $e_b$ has the form $e_b=(v_a,v_c)$, for some $v_a,v_c\in V$, and the vertices $v_a$ and $v_c$ are called the \textbf{endpoints} of $e_b$. A 
\textbf{walk} $W$ of length $m$ in the graph $G$ is an alternating sequence $W=v_0e_1v_1e_2\cdots v_{m-1}e_mv_m$ of vertices and edges such that $e_l=(v_{l-1},v_l)\in E$ for all $1\leq l\leq m$. The first vertex appearing in the alternating sequence, $v_0$, is called the \textbf{base point} of $W$. A walk $W$ is said to be \textbf{closed} if the two endpoints are the same, this is, when $v_0=v_m$. A closed walk $W$ is \textbf{backtrackless} if $e_l\neq e_{l+1}$ for all $l=1,2,\dots,m-1$. A backtrackless closed walk $W$ is \textbf{tailless} if $e_m\neq e_1$, and $W$ is called, in this case, a TBC walk. A \textbf{cycle} is a closed walk $W$ having distinct vertices and distinct edges, and if its alternating sequence has $k$ edges in it, then we call $W$ a $k$-cycle. The length of a shortest cycle is called the \textbf{girth} of the graph.


\subsection{CSS Codes}
For a positive integer $n$, the Pauli group on $n$ qubits $\cP_n$ is generated by $n$-fold Kronecker products of the Pauli matrices
\begin{equation}\label{e-pauli}
    {\sf I}_2,\quad \sfX = \begin{bmatrix}0&1\\1&0\end{bmatrix} ,\quad \sfZ = \begin{bmatrix}1&0\\0&-1\end{bmatrix},\quad {\sf Y} = i\sfX\sfZ.
\end{equation}
Then, a stabilizer group $\cS$ is a subgroup of $\cP_n$ that contains only commuting matrices but not $-{\sf I}_{2^n}$.
If we write a stabilizer $\cS$ in terms of its generators $\{s_1,\ldots,s_k\}$, then the corresponding quantum stabilizer code is the $2^{n-k}$ dimensional subspace of $\C^{2^n}$ consisting of all the states $|\psi\rangle$ that are stabilized by all of the generators $s_i$, that is $s_i|\psi\rangle = |\psi\rangle$.

It is convenient to think of the Pauli matrices in terms of the binary representation ${\sf I}_2 \leftrightarrow (0,0)$, ${\sf X} \leftrightarrow (1,0)$, ${\sf Z} \leftrightarrow (0,1)$, and ${\sf Y} \leftrightarrow (1,1)$. 
Then, any $n$-qubit Pauli matrix $P\in \cP_n$ can be represented as $p = (p_\sfX\mid p_\sfZ) \in \F_2^{2n}$ where $p_\sfX,p_\sfZ \in \F_2^n$ have ones at those positions where $P$ has $\sfX$ and $\sfZ$ respectively.
With this representation, a stabilizer group yields a full rank $k\times 2n$ matrix $H = (H_{\sf X}\mid H_{\sf Z})$.
Two stabilizers commute if and only if their respective representations $h = (h_{\sf X}\mid h_{\sf Z}), g = (g_{\sf X}\mid g_{\sf Z})$ of $H$ are orthogonal with respect to the symplectic inner product $h\odot g := h_{\sf X}g_{\sf Z}^\top + h_{\sf Z}g_{\sf X}^\top$.
Cumulatively, this leads to
\begin{equation}\label{e-sip}
H\odot H:=H_{\sf X}H_{\sf Z}^\top + H_{\sf Z}H_{\sf X}^\top = \boldzero.
\end{equation}

An important class of stabilizer codes are the Calderbank-Shor-Steane (CSS) codes~\cite{CS96,Steane96}, defined by a pair of classical linear codes $\cC_{\sf X}, \cC_{\sf Z}\subset \F_2^n$ such that $\cC_{\sf X}^\perp\subset \cC_{\sf Z}$. 
This condition forces two respective parity-check matrices $H_{\sf X}$ and $H_{\sf Z}$ to satisfy $H_{\sf Z}H_{\sf X}^\top = \boldzero$ and thus the matrix
\begin{equation}
H = \left[\!\!\begin{array}{c|c}H_{\sf X}&0\\0&H_{\sf Z} \end{array}\!\!\right]
\end{equation}
satisfies~\eqref{e-sip} and defines a stabilizer code. 
If $\cC_\sfX,\,\cC_\sfZ$ have parameters $[n, k_\sfX,d_\sfX]$, $[n,k_\sfZ,d_\sfZ]$ respectively, then the corresponding quantum code, which will be denoted $\cQ(\cC_\sfX,\cC_\sfZ)$, has parameters $[\![n,k_\sfZ+k_\sfX-n,\geq\min\{d_\sfZ,d_\sfX^\perp\}]\!]$ where $d_\sfX^\perp$ denotes the minimum distance of $\cC_\sfX^\perp$.

A quantum LDPC (QLDPC) code then is a quantum code with low-weight stabilizers or equivalently sparse $H$ matrix.

\section{Code Construction}
\label{sec:code_construction}

\subsection{Discussion of dyadic and quasi-dyadic codes}
Before introducing specific constructions, we discuss some general properties of dyadic and quasi-dyadic matrices along with their built-in duality, which in turn makes them appealing to the CSS construction.
\begin{theo}\label{T-space}
Given a dyadic matrix $D \in \cD_\ell$, if $ \sup(\s(D)) \subseteq \F_2^\ell$ is a subspace (or coset), then $\rk(D) = N/|\sup(\s(D))|$.
\end{theo}
\begin{proof}
Let $S:=\sup(\s(D))  = \sup(D_{{\bf 0},\bullet}) = \{\bx_1,\ldots,\bx_\o\}$ be a subspace of $\F_2^\ell$ of dimension $k$.
For the $\bx$th row of $D$, by definition, we have $d_{\bx,\by} = d_{{\bf 0},\bx+\by}$.
Thus we have that $\sup(D_{\bx,\bullet}) = \bx + \sup(D_{{\bf 0},\bullet})$; see also Proposition~\ref{P-permutation}(b). Since $S$ is a subspace, every row $D_{\bx,\bullet}$ of $D$ is either supported on $S$ if $\bx\in S$ or on the coset $\bx + S$ if $\bx\notin S$.
Of course $\F_2^N$ is a disjoint union of $2^{\ell-k} = N/|S|$ cosets of size $2^k$. 
If $S$ is subspace, then the first $2^k$ entries in the signature $\sigma(D)$ are ones. 
Thus, $D$ can be put in block-diagonal form ${\sf diag}(\boldone_{2^k},\cdots, \boldone_{2^k})$.
If $S$ is a coset, then by multiplying with a dyadic permutation as in Proposition~\ref{P-permutation}, the resulting dyadic matrix will have the same block-diagonal form.
Thus, there are $N/|S|$ blocks and each block has rank 1, and the result follows.
\end{proof}
\begin{cor}\label{C-dual}
Let $D \in \cD_\ell$. If $\wt(\s(D)) = 2$ then $\rk(D) = N/2$.
\end{cor}
\begin{exa}\label{Example}
Let $\ell = 3$ and $N = 2^\ell = 8$. There are $2^N = 256$ dyadic matrices. Half of them are full-rank, corresponding to signatures of weights $1,3,5$ and $7$; see Proposition~\ref{P-properties}(4). 
There are $\binom{8}{2}$ signatures of weight 2 that yield $28$ dyadic matrices of rank $4$. 
It is also easy to verify that for $\bv \notin \{{\bf 0}_N,{\bf 1}_N\}$, we have $\rs(D_\bv) = \rs(D_{\bv +{\bf 1}_N})$.
Thus, there are 28 additional dyadic matrices of rank 4 corresponding to signatures of weight 6.
Including the extreme cases $\boldone_N$ and $\boldzero_N$, this accounts for 186 dyadic matrices. The remaining 70 correspond to signatures of weight 4. 
Recall that there are $7$ subspaces of $\F_2^3$ of dimension $2$. Along with the additional coset, we have 14 signatures that yield dyadic matrices of rank $2$. The rest, that is the other 56, yield dyadic matrices of rank 4.
\end{exa}
\begin{rem}
It was shown in~\cite[Thm. 4.6]{mpk24} that if $D$ is a dyadic matrix with support $S$ such that $\sum_{\bx\in S}\bx \neq {\bf 0}$, then $\rk(D) = N/2$. 
Recall also that if $S$ is a subspace (or coset) of dimension at least 2, then $\sum_{\bx\in S}\bx = {\bf 0}$.
Thus, Theorem~\ref{T-space} is a partial converse of~\cite[Thm. 4.6]{mpk24}.
However, based on computation results and as also illustrated in Example~\ref{Example}, we conjecture that the full converse is true.
\end{rem}

Next, we consider the associated linear codes.

\begin{theo} \label{T:dyadic_code_properties}
Let $D\in \cD_\ell$ be a dyadic matrix. 
\begin{arabiclist}
    \item If $\wt(\s(D))$ is even, then $\cC= \rs(D)$ is self-orthogonal.
\item If 
$\sup(\s(D))$ is a subspace or a coset of $\F_2^\ell$ of size $2^k$, then $\cC = \rs(D)$ is an $[N,2^{\ell-k},2^k]$ code whereas $\cC^\perp$ is an $[N,N-2^{\ell-k},2]$ code.
\end{arabiclist}
\end{theo}
\begin{proof}
1) Let $D\in\cD_\ell$ be a dyadic matrix such that $\wt(\s(D))$ is even. Then, $D^2 = \boldzero_N$ and thus $\rs(D) \subseteq \ker(D)$.
Hence, $\cC^\perp = \ker(D)$ and the code $\cC$ is self-orthogonal.

2) The statement of the dimensions follows directly by Theorem~\ref{T-space}. The statement on the minimum distance of $\cC$ follows by the fact that, as also highlighted in the proof of the same Theorem, $D$ can be permuted to a block diagonal form where each block is $\boldone_{2^k}$. Row permutations of course do not change the row space and column permutations do not change the weight. Therefore, $d(\cC) = 2^k$. 
For $\cC^\perp$, by the same argument, $D$ has (at least) two equal columns. Thus $d(\cC^\perp) = 2$.
\end{proof}
In general, if $H_1$ and $H_2$ are two matrices such that 
\begin{equation}\label{e-CSS}
H_1 H_2^\top = \boldzero    
\end{equation} 
then the codes $\cC_1 = \ker(H_1)$ and $\cC_2 = \ker(H_2)$ satisfy the CSS condition $\cC_2^\perp \subseteq \cC_1$. 
This condition is easily met with dyadic matrices due to their built-in duality, providing a source for CSS codes.
\begin{lem}
Let $D_\bu, D_\bv\in\cD_n$ be two dyadic matrices. Then $D_\bu\cdot D_\bv = \boldzero_N$ if and only if $\bu \in \ker(D_\bv)$ if and only if $\bv\in \ker(D_\bu)$.
\end{lem}
\begin{proof}
Note that $\s(D_\bu\cdot D_\bv) = \bu D_\bv$.
Thus, $D_\bu\cdot D_\bv = \boldzero_N$ if and only if $\bu D_\bv = \s(D_\bu\cdot D_\bv) = {\bf 0}_N$ if and only if $\bu\in \ker(D_\bv)$.
Since $D_\bu\cdot D_\bv = D_\bv\cdot D_\bu$, everything is equivalent with $\bv\in\ker(D_\bu)$.
\end{proof}
\begin{exa}\label{exa}
Let $\bv\in \F_2^{16}$ be such that $\sup(\bv) = \{1,5,6,7,11,14\}$. Then $\rk(D_\bv) = 6$ and $\cC_1 = \ker(D_\bv)$ is a $[16,10]$ code.
Choosing $\bu\in\cC_1$ with $\sup(\bu) = \{4,5,6,8,10,13\}$, we have that $\rk(D_\bu) = 6$ and $\cC_2 = \ker(D_\bu)$ is also a $[16,10]$ code such that $\cC_2^\perp \subseteq \cC_1$. Direct computation yields $d(\cC_1) = d(\cC_2) = 4$ and thus $\cQ(\cC_1,\cC_2^\perp)$ is a $[\![16,4,4]\!]$ CSS code.
\end{exa}

Choosing $H_1 = H_2 =D$ to be a dyadic matrix such that $\wt(\s(D))$ is even, then condition~\eqref{e-CSS} is satisfied automatically (recall that dyadic matrices are symmetric), and as described above, we obtain CSS codes.
Another construction is to choose 
\begin{equation}\label{e-bicycle}
H_1 = \matrix{A&B},\quad H_2 = \matrix{B&A}
\end{equation}
where $AB = BA$, in which case $H_1H_2^\top=\boldzero$. 
One example of such method is the class of \textbf{generalized bicycle codes}~\cite{panteleev2021degenerate, kovalev2013quantum}, where the matrices $A$ and $B$ are chosen to be (versions of) circulant matrices so that the commutativity condition is automatically satisfied. Yet again, because dyadic matrices commute, they can be used in such a construction.
Moreover, if $A$ and $B$ are dyadic matrices, then $H_1$ and $H_2$ are halves of another dyadic matrix.
\begin{exa}
Let $D_1,D_2 \in \cD_\ell$ be such that $\rk(D_1) = \rk(D_2) =  N/2$. Then, $\cC_i = \ker(D_i)$ is self-dual, which in turn yields an $[\![N,0,d(\cC_i)]\!]$ CSS code. 
In~\eqref{e-bicycle}, let us take $H_1 = \matrix{D_1&D_2}, H_2 = \matrix{D_2&D_1}$. 
Let us now consider the CSS pair $(\widehat{\cC}_1,\widehat{\cC}_2)$ where $\widehat{\cC}_i = \ker(H_i)$. 
Note that 
\[
\dim(\widehat{\cC}_i) = N + \dim(\cC_1\cap\cC_2),
\]
and thus, such CSS pair will yield an $[\![2N, K, d]\!]$ CSS code where
\[
K = 2\dim(\cC_1\cap\cC_2), \quad d\geq \min\{d(\cC_1),d(\cC_2)\}.
\]
In such a way, it is possible to construct a $[\![64,16,8]\!]$ CSS code.
\end{exa}

\subsection{Codes from quasi-dyadic matrices}

The {\bf hypergraph product codes}~\cite{tillich2013quantum} are a class of CSS codes constructed from two matrices $A\in \F_2^{m_A\times n_A},B\in \F_2^{m_B\times n_B}$ and 
\begin{align*}
H_1 & := \matrix{A\otimes I_{m_B}&I_{m_A}\otimes B},\\
H_2 & := \matrix{I_{n_A}\otimes B^\top&A^\top\otimes I_{n_B}},
\end{align*}
which straightforwardly satisfy~\eqref{e-CSS}. In fact,~\eqref{e-CSS} is again satisfied if the matrices $A,B$ take entries from any commutative ring. The {\bf lifted product codes}~\cite{Panteleev_2022} are obtained from the hypergraph product construction using the commutative ring of circulant matrices. Since dyadic matrices form a commutative ring as well, they can be used in such construction. 
Note that an $c\times v$ matrix over the ring $\cD_\ell$ is simply an $cN\times v N$ block matrix where each block is an $N\times N$ dyadic matrix, that is, a  quasi-dyadic matrix. We will denote $\cD_\ell^{c\times v}$ the ring of such matrices.
\begin{exa}
With quasi-dyadic matrices in $\cD_2^{3\times 4}$ it is possible to further improve upon Example~\ref{exa}. Indeed, there exists a $[\![16,6,4]\!]$ CSS code, which is the highest achievable distance for the given length and dimension~\cite{Grassl:codetables}. Similarly, quasi-dyadic matrices in $\cD_2^{3\times 5}$ can be used to construct a $[\![20,6,4]\!]$ CSS code.
\end{exa}

We now give the main construction---a generalization of~\cite{kasai}---that uses permutation dyadic matrices.
\begin{construction}\label{con-main}
Fix a positive integer $\o$ and $\bx_i,\by_i \in \F_2^\ell$ for $i = 1,\ldots,\o$. 
Let $\s, \tau$ be two permutations of $[\o]$ of order $\o$, and consider the $\o\times 2\o$ block quasi-dyadic matrices $H_\sfX = \left[H_{\sfX, L} \mid H_{\sfX, R}\right]$ and $H_\sfZ = \left[H_{\sfZ, L} \mid H_{\sfZ, R}\right]$ where
\begin{align*}
H_{\sfX, L} = \left[P_{\bx_{\s^{(i-1)}(j)}}\right]_{i,j\in[\o]}, \,\, & H_{\sfX, R} = \left[P_{\by_{\s^{(i-1)}(j)}}\right]_{i,j\in[\o]},\\
H_{\sfZ, L} = \left[P_{\by_{\tau^{(i-1)}(j)}}\right]_{i,j\in[\o]}, \,\, & H_{\sfZ, R} = \left[P_{\bx_{\tau^{(i-1)}(j)}}\right]_{i,j\in[\o]},
\end{align*}
and each of the blocks is a dyadic permutation matrix.
\end{construction}

In Construction~\ref{con-main}, we choose the permutations $\s$ and $\tau$ such that $H_\sfX H_\sfZ^\top = \boldzero$.
For this, we must have 
\begin{equation}\label{e-CSScondition}
\sum_{k=1}^\o P_{\bx_{\s^{(i-1)}(k)}} P_{\by_{\tau^{(j-1)}(k)}} = \sum_{k=1}^\o P_{\by_{\s^{(i-1)}(k)}} P_{\bx_{\tau^{(j-1)}(k)}}
\end{equation}
for all $i,j\in [\o]$.
Since the product of dyadic permutation matrices is again a dyadic permutation matrix (see Proposition~\ref{P-permutation}) and since 
every dyadic matrix is a unique sum of dyadic permutations
(see also~\eqref{parity_check_matrix_single_dyadics} below), we must guarantee that every term on the left-hand side of~\eqref{e-CSScondition} must appear on the right-hand side of~\eqref{e-CSScondition} and vice versa.
Such condition can be easily achieved by choosing $\s,\tau$ to be cyclic permutations. 
For simulations, $\s,\tau$ are also optimized so that the resulting codes have designated girth and small number of short cycles. 

We will compare our Construction~\ref{con-main} with the recent Construction B of~\cite{BBS26}, which we lay out next.
\begin{construction}\label{con-BBS}
Fix an even integer $u = 2v$ and $\bx_1,\ldots,\bx_v\in \F_2^\ell$. 
Consider the quasi-dyadic matrix
\begin{equation}
H_P = \matrix{P&P_{\bx_1}&P&P_{\bx_2}&\cdots&P&P_{\bx_v}\\P_{\bx_v}&P&P_{\bx_1}&P&\cdots & P_{\bx_{v-1}}&P\\P&P_{\bx_v}&P&P_{\bx_{v-1}}&\cdots&P&P_{\bx_1}\\P_{\bx_1}&P&P_{\bx_v}&P&\cdots & P_{\bx_{2}}&P}
\end{equation}
where $P$ is a fixed dyadic permutation and $P_{\bx_i}$'s are the corresponding dyadic permutations.
\end{construction}
\begin{rem}
In Construction~\ref{con-BBS}, the arrangement of the dyadic permutations $P$ and $P_{\bx_i}$ is chosen so that $H_PH_P^\top = \boldzero$ and thus yielding a dual-containing CSS code. However, by construction (see~\cite[Thm.~5.1]{mpk24} or~\cite[Cor.~2]{BBS26} for instance), the associated Tanner graph will inevitably have girth 4.
As we will see from the simulations and as also discussed below, this is a major drawback in terms of performance.
Nevertheless, using our cycle analysis, the choices of $P$ and $P_{\bx_i}$ can be optimized (instead of random as in~\cite{BBS26}) so that the associated Tanner graph has a minimal number of 4-cycles, and this alone gives significant performance gain.
\end{rem}


\section{Girth Analysis}
\label{sec:girth_analysis}

In this section, we present the girth and cycle analysis of dyadic and quasi-dyadic codes. Specifically, in Sections \ref{sec:counting_cycles_single_dyadics} and \ref{sec:counting_cycles_dyadics}, we propose strategies to count 4-cycles in the Tanner graph of dyadic and quasi-dyadic codes, respectively; the rest of the analysis is presented in Appendix \ref{appendix:girth_analysis}. 
First, we present the background required to establish a connection between the cycles in the Tanner graph and the TBC walks in the protograph.

{

\subsection{Connection Between TBC Walks in the Protograph and Cycles in the Tanner Graph}
\label{sec:counting_cycles_dyadics_background}
 
Let $H$ be the parity-check matrix of a dyadic or quasi-dyadic code and let $G=(V,E)$ be the corresponding Tanner graph. The \textbf{adjacency matrix} $A=(A_{ij})$ is the symmetric binary matrix with $A_{ij}=1$ if $(v_i,v_j)\in E$, and $A_{ij}=0$ otherwise. After some reordering of the vertices, if necessary, we can write $A$, for either the scalar or block representation of $H$, in the compact expression
\begin{equation}
\label{counting_cycles_adjacencymatrix}
A=\matrix{0&H \\ H^{\mathsf{T}}&0}.
\end{equation}
The powers of $A$, and in particular the matrices
\begin{equation*}
\label{counting_cycles_BnH}
B_t(H)=\left(HH^\mathsf{T}\right)^{\lfloor{t/2}\rfloor}H^{(t\mod 2)}, \quad t\geq 0,
\end{equation*}
give information about the walks in the Tanner graph \cite{wyz08}. It is not difficult to see that, for any nonnegative integer $t$, we have 
\begin{equation}
\label{counting_cycles_adjacencymatrix_powers}
A^{2t} = \matrix{B_{2t}(H)&0 \\ 0&B_{2t}(H^\mathsf{T})} \quad\text{and}\quad A^{2t+1}=\matrix{0&B_{2t+1}(H) \\ B_{2t+1}(H^\mathsf{T})&0}.
\end{equation}

\begin{theorem}[\hspace{-1sp}{\cite{mw03}}]
\label{counting_cycles_powerofadjacencymatrix}
Given an adjacency matrix $A$, the entry ${(A^m)}_{ij}$ is equal to the number of walks of length $m$ between the vertices $v_i$ and $v_j$.
\end{theorem}

Similarly to the cycle analysis of QC-LDPC codes presented in \cite{gfsm23c}, the block representation of quasi-dyadic LDPC codes allows for a reduction in the complexity of our computations. 
The $k$th power of the scalar adjacency matrix $A$ of the Tanner graph can be used to determine the number of $k$-walks between any two vertices, as we have seen in Theorem \ref{counting_cycles_powerofadjacencymatrix}. The $k$th power of the block version of the adjacency matrix can be used to describe the edges traversed in a $k$-walk between any two vertices in the protograph. For example, if $A$ is the block version of the adjacency matrix in \eqref{counting_cycles_adjacencymatrix}, then
\begin{equation*}
\label{counting_cycles_dyadics_polynomialsquareofA}
{(A^2)}_{ij}=\displaystyle \sum_{l=0}^{n_c+n_v-1}A_{il}A_{lj},
\end{equation*}
and every term of ${(A^2)}_{ij}$ is a product of the form $P_{\bc_{il}}P_{\bc_{lj}}=P_{\bc_{il}+\bc_{lj}}$, where $P_{\bc_{il}}$ and $P_{\bc_{lj}}$ come from $A_{il}$ and $A_{lj}$, respectively. Each one of the two dyadic permutation matrices $P_{\bc_{il}}$ and $P_{\bc_{lj}}$ corresponds to a unique edge in the protograph, and the order in which they appear in the product is the order used to traverse the walk in the protograph. The index $\bc_{il}+\bc_{lj}$, in consequence, corresponds to the two edges traversed from vertex $v_i$ to vertex $v_l$ to vertex $v_j$ in the protograph. In the same way, every term of ${(A^3)}_{ij}$ is a product of the form $P_{\bc_{il}}P_{\bc_{lk}}P_{\bc_{kj}}=P_{\bc_{il}+\bc_{lk}+\bc_{kj}}$, and the index $\bc_{il}+\bc_{lk}+\bc_{kj}$ corresponds to the three edges traversed in the protograph from vertex $v_i$ to vertex $v_l$ to vertex $v_k$ to vertex $v_j$. In general, every term of ${(A^m)}_{ij}$ is of the form $P_{\bc_{il_1}}P_{\bc_{l_1l_2}}\cdots P_{\bc_{l_mj}}=P_{\bc_{il_1}+\bc_{l_1l_2}+\cdots+\bc_{l_mj}}$, and each one of them corresponds to a walk of length $m$ and the specific order in which it is traversed.

\begin{theorem}[\hspace{-1sp}{\cite{gfsm23c}}]
\label{counting_cycles_dyadics_polynomialpowerofadjacencymatrix}
Given an adjacency matrix $A$, every term of ${(A^m)}_{ij}$ is of the form $P_{\bc_{il_1}}P_{\bc_{l_1l_2}}\cdots P_{\bc_{l_mj}}$ and corresponds to a walk of length $m$ between the vertices $v_i$ and $v_j$ in the protograph.
\end{theorem}

\begin{definition}
\label{counting_cycles_dyadics_permutationshift}
The index $\bc_{il_1}+\bc_{l_1l_2}+\cdots+\bc_{l_mj}$ corresponding to the product $P_{\bc_{il_1}}P_{\bc_{l_1l_2}}\cdots P_{\bc_{l_mj}}$ in Theorem \ref{counting_cycles_dyadics_polynomialpowerofadjacencymatrix} is called a \textbf{permutation shift}. 
\end{definition}

Suppose that $P_{\bc_{il_1}}P_{\bc_{l_1l_2}}\cdots P_{\bc_{l_mj}}$ and $P_{\bc_{il_1}'}P_{\bc_{l_1l_2}'}\cdots P_{\bc_{l_mj}'}$ are two terms of ${(A^m)}_{ij}$ describing two $m$-walks between vertices $v_i$ and $v_j$ in the protograph. Then the combination
\begin{equation*}
P_{\bc_{il_1}}P_{\bc_{l_1l_2}}\cdots P_{\bc_{l_mj}}P_{\bc_{l_mj}'}\cdots P_{\bc_{l_1l_2}'}P_{\bc_{il_1}'}
\end{equation*}
of the first walk and the reversal of the second one describes a closed $(2m)$-walk that starts and ends at the vertex $v_i$, and that has the vertex $v_j$ midway. Hence, the entries ${(A^m)}_{ij}$ of the power $A^m$ describe all the $m$-walks in the protograph and can be used to count certain cycles in the Tanner graph.

Let $G=(V,E)$ be a protograph described by matrix $B={(b_{ij})}_{n_c\times n_v}$. Each row and each column of $B$ corresponds to a check node and a variable node in the protograph, respectively. Once a lifting factor $N=2^\ell$ is chosen, for each vertex $v\in V$ in the protograph, either a check node or a variable node, we create $N$ copies of it and denote them by $\tilde{v}^l$, for $l\in[N]$. For each edge $e=(u,v)\in E$ in the protograph, there is, associated with it, a dyadic permutation matrix $P_\ba$, where $a\in[N]$. Once the value for $a$ is chosen, we create $N$ copies of $e$ and denote them by $\tilde{e}^l$, for $l\in[N]$. The vertices that are endpoints of these edges are permuted in such a way that we have $\tilde{e}^l=(\tilde{u}^l,\tilde{v}^{s_{l,a}})$, where $s_{l,a}$ is the integer corresponding to $\bl+\ba$. 
If we let $\tilde{V}=\left\{\tilde{v}^l \mid v\in V, l\in[N] \right\}$ and $\tilde{E}=\left\{\tilde{e}^l \mid e\in E, l\in[N] \right\}$, then the graph $\tilde{G}=(\tilde{V},\tilde{E})$ is an $N$-fold graph cover, or \textbf{lifting}, of the protograph, and we call it the Tanner graph. The process of creating the $N$ copies $\tilde{v}^l$ of the vertex $v$ and the $N$ copies $\tilde{e}^l$ of the edge $e$ induces a projection map $p:\tilde{G}\rightarrow G$, and we call $p$ the natural projection map. The set of vertices $\left\{\tilde{v}^l \mid l\in[N] \right\}$ and the set of edges $\left\{\tilde{e}^l \mid l\in[N] \right\}$, denoted by $p^{-1}(v)$ and $p^{-1}(e)$, respectively, are called the fiber over the vertex $v$ and the fiber over the edge $e$, respectively, under the natural projection map.

All the results that follow have been slightly adapted to the dyadic case. Lemmas \ref{counting_cycles_dyadics_lemma_permutation_shift} and \ref{counting_cycles_dyadics_imagecycle}, based on some results from \cite{gt87}, are useful for studying both the images of cycles in the Tanner graph and the preimages of TBC walks in the protograph.

\begin{lemma}[\hspace{-1sp}\cite{abaa11a,abaa11b}]
\label{counting_cycles_dyadics_lemma_permutation_shift}
Let $\tilde{G}$ be an $N$-fold graph cover of the protograph $G$. Let $W$ be a $k$-walk in $G$ starting at vertex $v$ and ending at vertex $v'$, and having edge sequence $e_1, e_2,\dots,e_k$ with associated dyadic permutation matrices $P_{\bs_1},P_{\bs_2},\dots,P_{\bs_k}$. Then the permutation shift $\bs$ that maps $\tilde{v}$, the inverse image of $v$ in $\tilde{G}$, to $\tilde{v'}$, the inverse image of $v'$ in $\tilde{G}$, through the walk $\tilde{W}$ is given by $\bs_1+\bs_2+\cdots+\bs_k$. 
\end{lemma}

\begin{remark}
If the walk $W$ in Lemma \ref{counting_cycles_dyadics_lemma_permutation_shift} is traversed in the opposite direction starting at vertex $v'$ and ending at vertex $v$, then its permutation shift is given by $\bs'=-\bs=\bs$. 
\end{remark}

We denote by $\mathbb{Z}_t$ the additive group of integers modulo $t$. For any element $a\in\mathbb{Z}_t$, the \textbf{order} of $a$ is the smallest positive integer $m$ such that $m\cdot a=0$.

\begin{lemma}[\hspace{-1sp}\cite{abaa11a}]\label{counting_cycles_dyadics_imagecycle}
Let $\tilde{G}$ be an $N$-fold graph cover of the protograph $G$ and let $\tilde{W}$ be a $k$-cycle in $\tilde{G}$. Then $\tilde{W}$ is projected onto a TBC walk $W$ of length $k/m$, where $m\geq 1$ is the order of the permutation shift of $W$ in $\mathbb{F}_2^\ell$.
\end{lemma}

\begin{remark}
The order of a TBC walk $W$, which is referred to as the order of its permutation shift $\bs$ in the previous lemmas when considered as an element of $\mathbb{F}_2^\ell$, is given by $2/\gcd{(2,s)}$, where $\bs$ is as in Lemma \ref{counting_cycles_dyadics_lemma_permutation_shift} and $\gcd$ denotes the greatest common divisor. 
\end{remark}

We combine the following lemma with our analysis of TBC walks to count cycles in the Tanner graph.
\begin{lemma}[\hspace{-1sp}\cite{kel08}]
\label{counting_cycles_dyadics_preimagewalk}
Let $\tilde{G}$ be an $N$-fold graph cover of the protograph $G$ and let $W$ be a closed $k$-walk in $G$. Then the inverse image of $W$ in $\tilde{G}$ is the union of $N/m$ closed $(km)$-walks, where $m\geq 1$ is the order of the permutation shift of $W$ in $\mathbb{F}_2^\ell$.
\end{lemma}

In \cite{gfsm23c}, the authors extended Lemma \ref{counting_cycles_dyadics_preimagewalk}, stated for closed $k$-walks, to TBC walks of length $k$.

\begin{theorem}[\hspace{-1sp}\cite{gfsm23c}]
\label{counting_cycles_dyadics_preimageTBCwalk}
Let $\tilde{G}$ be an $N$-fold graph cover of the protograph $G$ and let $W$ be a TBC walk of length $k$ in $G$. Then the inverse image of $W$ in $\tilde{G}$ is the union of $N/m$ TBC walks of length $km$, where $m\geq 1$ is the order of the permutation shift of $W$ in $\mathbb{F}_2^\ell$.
\end{theorem}

The following lemma explains why we restrict our analysis to $k$-cycles with $k<2g$, where $g$ is the girth of the Tanner graph.

\begin{lemma}[\hspace{-1sp}{\cite{kb13_b,abaa11a}}]
\label{counting_cycles_TBCwalks_equal_cycles}
Let $G$ be a graph with girth $g$. Then the set of TBC walks of length $k$ coincides with the set of $k$-cycles if $k<2g$.
\end{lemma}

\begin{corollary}
\label{counting_cycles_dyadics_preimageTBCwalk_are_cycles}
Let $\tilde{G}$ be an $N$-fold graph cover of the protograph $G$ and let $W$ be a TBC walk of length $k$ in $G$. Then the inverse image of $W$ in $\tilde{G}$ is the union of $N/m$ $(km)$-cycles, where $m\geq 1$ is the order of the permutation shift of $W$ in $\mathbb{F}_2^\ell$.
\end{corollary}

\begin{remark}
\label{counting_cycles_TBC_and_cycles}
As a direct consequence of Lemmas \ref{counting_cycles_dyadics_imagecycle} and \ref{counting_cycles_TBCwalks_equal_cycles}, Theorem \ref{counting_cycles_dyadics_preimageTBCwalk}, and Corollary \ref{counting_cycles_dyadics_preimageTBCwalk_are_cycles}, the TBC walks in the protograph of a quasi-dyadic LDPC code are the necessary and sufficient structures needed to describe all the $k$-cycles, $k<2g$, in the Tanner graph.
\end{remark}

We define the equivalence of closed walks.

\begin{definition}
\label{equivalent_closed_walks}
Two closed walks $W_1$ and $W_2$ are said to be \textbf{equivalent} if one can be obtained from the other by a change of base point, a change in direction, or both.
\end{definition}

\begin{remark}
If $W$ is a closed walk of length $k$, then there are $2k$ equivalent closed walks to $W$.
\end{remark}}

\subsection{Cycles in Dyadic Codes}
\label{sec:counting_cycles_single_dyadics}

Let $D=(D_{\bx,\by})$ be an $N\times N$ dyadic matrix over $\mathbb{F}_2$, where $N=2^\ell$. The integer $\ell$  will represent the degree of the graph lift. Suppose that $\sigma(D)$ has weight $\omega$ and that ${\sf supp}(\sigma(D))=\{x_0,x_1,x_2,\dots,x_{\omega-1}\}\subseteq[N]$. Then we can write $D$ as a summation of $\omega$ dyadic permutation matrices
\begin{equation}
\label{parity_check_matrix_single_dyadics}
D=P_{\bx_0}+P_{\bx_1}+P_{\bx_2}+\cdots+P_{\bx_{\omega-1}}.
\end{equation}
If $\omega=1$, then $D=P_{\bx_0}$, so its Tanner graph is acyclic. If $\omega\geq2$, then the girth of $D$ is 4. To see why this is the case, first notice that $\{x_0,x_1\}\subseteq{\sf supp}(\sigma(D))$. Let $\bc=\bx_0+\bx_1$.
By assumption, $D_{\bzero,\bx_0}=D_{\bzero,\bx_1}=1$. Applying Definition \ref{definition_dyadic_matrix}, we see that $D_{\bc,\bx_0} = D_{\bzero,\bc+\bx_0} = D_{\bzero,\bx_0+\bx_1+\bx_0} = D_{\bzero,2\bx_0+\bx_1} = D_{\bzero,\bx_1} = 1$ and $D_{\bc,\bx_1} = D_{\bzero,\bc+\bx_1} = D_{\bzero,\bx_0+\bx_1+\bx_1} = D_{\bzero,\bx_0+2\bx_1} = D_{\bzero,\bx_0} = 1$. Hence, $D$ has a $2\times2$ submatrix with all its entries equal to 1, so it contains at least one 4-cycle. 

\subsubsection{Counting 4-cycles}
\label{sec:counting_cycles_counting_4cycles_single_dyadics}

If $D$ is a dyadic matrix as in \eqref{parity_check_matrix_single_dyadics}, then we can consider this matrix as a lifting of the $1\times1$ protograph with unique entry equal to $\omega$. By the results in Section \ref{sec:counting_cycles_dyadics_background}, and in particular by Lemma \ref{counting_cycles_dyadics_imagecycle}, the 4-cycles in the Tanner graph are projected onto either TBC walks of length 2 that are traversed twice or TBC walks of length 4, in the protograph, and these walks can be studied by computing the product $DD^\mathsf{T}$ according to Theorem \ref{counting_cycles_dyadics_polynomialpowerofadjacencymatrix}. A TBC walk of length 4 in the protograph has the form $[\bx_{u},\bx_{v},\bx_{v'},\bx_{u'}]$ and counting the number of 4-cycles requires analyzing the mset $A=\left\{\left(u,v,\bx_u+\bx_{v}\right)\mid u,v\in[\omega],u\neq v\right\}$. Consider the following theorem.
\begin{theorem}
\label{counting_cycles_thm_4cycles_ncxnv_single_dyadics}
Let $D$ be as in \eqref{parity_check_matrix_single_dyadics} and let
\begin{equation*}
A=\left\{\left(u,v,\bx_u+\bx_{v}\right)\mid u,v\in[\omega], u\neq v\right\}.
\end{equation*}
\noindent For $u,v,u',v'\in[\omega]$, let $(u,v,\balpha_{u,v}),(u',v',\balpha_{u',v'})\in A$ be such that $\balpha_{u,v}=\balpha_{u',v'}$. 
Then this repetition $\balpha_{u,v}=\balpha_{u',v'}$ lifts to a collection of 4-cycles in the Tanner graph if $u\neq u'$ and $v\neq v'$. The total number of 4-cycles in the Tanner graph, $\mathcal{N}_4$, is given by
\begin{equation}
\label{counting_cycles_formula_num4cycles_ncxnv_single_dyadics}
\mathcal{N}_4=\frac{2^\ell}{2}\cdot\frac{1}{2}\cdot\mathcal{R}_{A}^{*c}+2^\ell\cdot\frac{1}{4}\cdot\mathcal{R}_{A}^{*nc} = \frac{2^\ell}{4}\cdot\left(\mathcal{R}_{A}^{*c}+\mathcal{R}_{A}^{*nc}\right),
\end{equation}
where $\mathcal{R}_{A}^{*c}$ and $\mathcal{R}_{A}^{*nc}$ are the numbers of repetitions $\balpha_{u,v}=\balpha_{u',v'}$ in the mset $A$ satisfying $u=v'$ and  $v\neq u'$.
\end{theorem}
\begin{proof}
For $u,v,u',v'\in[\omega]$, let $(u,v,\balpha_{u,v}),(u',v',\balpha_{u',v'})\in A$ be such that $\balpha_{u,v}=\balpha_{u',v'}$. The walk of length 4 corresponding to this repetition is given by $W=[\bx_u,\bx_v,\bx_{v'},\bx_{u'}]$. If $u\neq u'$ and $v\neq v'$, then this walk is actually a TBC walk with permutation shift $\balpha_{u,v}+\balpha_{u',v'}=\bzero$, so by Corollary \ref{counting_cycles_dyadics_preimageTBCwalk_are_cycles}, its inverse image is the union of $N/m$ 4-cycles, where $m\geq 1$ is the order of the permutation shift of $W$ in $\mathbb{F}_2^\ell$. Notice that $m\in\{1,2\}$, leading to two cases. 

If $m=1$, then all the edges in $W$ are distinct. For simplicity, let us focus on the indices $(u,v,v',u')$. When computing all the walks equivalent to $W$, we obtain exactly 4 distinct elements with indices $(u,v,v',u'),(v,v',u',u),(v',u',u,v)$, and $(u',u,v,v')$, which means that the repetition $\balpha_{u,v}=\balpha_{u',v'}$ is counted four times in the construction of $A$. The number of 4-cycles arising from this repetition is $2^\ell$. 

If $m=2$, then $W$ is obtained by traversing a TBC walk of length 2 twice. When computing the walks equivalent to $W$, we obtain exactly 2 distinct elements with indices $(u,v,u,v)$ and $(v,u,v,u)$ because the conditions $u=v'$ and $v=u'$ are required, which means that the repetition $\balpha_{u,v}=\balpha_{u',v'}$ is counted twice in the construction of $A$. The number of 4-cycles arising from this repetition is $\frac{2^\ell}{2}$. Then equation \eqref{counting_cycles_formula_num4cycles_ncxnv_single_dyadics} follows and we conclude the proof. 
\end{proof}

\begin{example}
\label{counting_cycles_example_4cycles_g4_single_dyadics}
Let $D$ be the parity-check matrix given by $D=P_{\boldsymbol0}+P_{\boldsymbol1}+P_{\boldsymbol2}+P_{\boldsymbol3}$, and let $N=2^3$. Since the weight of $\sigma(D)$ is $\omega=4\geq2$, the Tanner graph has 4-cycles. To study the 4-cycles in the Tanner graph, we construct the mset in Theorem \ref{counting_cycles_thm_4cycles_ncxnv_single_dyadics}:  
\begin{align*}
A=\{&(0, 1, \mathbf{1}), (0, 2, \mathbf{2}), (0, 3, \mathbf{3}), (1, 0, \mathbf{1}), (1, 2, \mathbf{3}), (1, 3, \mathbf{2}), \\ & (2, 0, \mathbf{2}), (2, 1, \mathbf{3}), (2, 3, \mathbf{1}), (3, 0, \mathbf{3}), (3, 1, \mathbf{2}), (3, 2, \mathbf{1})\}.
\end{align*}
Further computations show that $\mathcal{R}_{A}^{*c}=12$ and $\mathcal{R}_{A}^{*nc}=24$. According to \eqref{counting_cycles_formula_num4cycles_ncxnv_single_dyadics}, the total number of 4-cycles in the Tanner graph of $D$ is
\begin{align*}
\mathcal{N}_4 &= \frac{2^\ell}{4}\cdot\left(\mathcal{R}_{A}^{*c}+\mathcal{R}_{A}^{*nc}\right) = \frac{2^3}{4}\cdot\left(12+24\right) = 72.
\end{align*}
The parity-check matrix $D$ gives a code with parameters $[8,6,2]$, meaning it is almost MDS.
\end{example}

\begin{remark}
\label{remark_N4_lower_bound}
If $\omega\geq2$, then $W=[\bx_0,\bx_1,\bx_0,\bx_1]$ is a TBC walk of length 4 obtained by traversing a TBC walk of length 2 twice. This TBC walk always have permutation shift $\bzero$ in $\mathbb{F}_2^\ell$. By the proof of Theorem \ref{counting_cycles_thm_4cycles_ncxnv_single_dyadics}, notice that $\mathcal{N}_4\geq\frac{2^\ell}{2}\cdot\frac{1}{2}\cdot\omega(\omega-1)$. The first term of \eqref{counting_cycles_formula_num4cycles_ncxnv_single_dyadics} is equal to the right-hand side of this inequality since that is the number of ways we can construct the TBC walks of length 2 that are traversed twice. If $\omega\in\{0,1\}$, we naturally have $\mathcal{N}_4=0$.  
\end{remark}

If $D$ is a dyadic matrix as in \eqref{parity_check_matrix_single_dyadics} and $\sigma(D)$ has weight $\omega\geq2$, then $D$ has girth 4 by Remark \ref{remark_N4_lower_bound}. We can use our cycle analysis to also compute the number of 6-cycles and this is presented in Appendix \ref{sec:counting_cycles_single_dyadics_continuation}.

\subsection{Cycles in Protograph-based LDPC Codes from Dyadics}
\label{sec:counting_cycles_dyadics}
 
Let $\mathcal{C}$ be a protograph-based LDPC code with parity-check matrix $H$ given by 
\begin{equation}
\label{parity_check_matrix_dyadics}
H=\matrix{ P_{\bx_{0,0}}&P_{\bx_{0,1}}&\cdots&P_{\bx_{0,n_v-1}} \\ P_{\bx_{1,0}}&P_{\bx_{1,1}}&\cdots&P_{\bx_{1,n_v-1}} \\ \vdots&\vdots&\ddots&\vdots \\ P_{\bx_{n_c-1,0}}&P_{\bx_{n_c-1,1}}&\cdots&P_{\bx_{n_c-1,n_v-1}}}.
\end{equation}

\subsubsection{Counting 4-cycles}
\label{counting_cycles_counting_4cycles_dyadics}

We have seen that the 4-cycles in the Tanner graph are projected onto either TBC walks of length 2 that are traversed twice or TBC walks of length 4, in the protograph, and these walks can be studied by computing the product $HH^\mathsf{T}$ according to Theorem \ref{counting_cycles_dyadics_polynomialpowerofadjacencymatrix}. Since Tanner graphs are bipartite, a 4-cycle can only be projected onto a TBC walk of length 4. Because of this, the 4-cycles are projected onto a $2\times2$ submatrix of the form
\begin{equation*}
H_2=\matrix{P_{\bh_u}&P_{\bh_v} \\ P_{\bi_u}&P_{\bi_v}}.
\end{equation*}
In terms of notation, an index of the form $\bh_u$ indicates that the dyadic permutation matrix $P_{\bh_u}$ appears in the $h$th row and $u$th column of $H$. The corresponding TBC walk of length 4 in the protograph,
$[\bh_u,\bi_u,\bi_v,\bh_v]$, gives the product
\begin{equation*}
{\left(P_{\bh_v}\right)}^{-1}\cdot P_{\bi_v}\cdot{\left(P_{\bi_u}\right)}^{-1}\cdot P_{\bh_u} = P_{\bh_v}\cdot P_{\bi_v}\cdot P_{\bi_u}\cdot P_{\bh_u} = P_{\bh_u+\bi_u+\bi_v+\bh_v},
\end{equation*}
so we require that 
$\bh_u+\bi_u+\bi_v+\bh_v\neq\bzero$ to avoid 4-cycles in the Tanner graph. This is the case since the only dyadic permutation matrix that has fixed points is the identity matrix, so we need to make sure that the left-hand side is not equal to $\bzero$.  
\begin{theorem}
\label{counting_cycles_thm_4cycles_ncxnv_dyadics}
Let $H$ be as in \eqref{parity_check_matrix_dyadics}. A repetition in any of the following $\binom{n_c}{2}$  
msets
\begin{equation*}
A_{h,i}=\left\{\bh_u+\bi_u\mid u\in[n_v]\right\}, \quad h<i,
\end{equation*}
lifts to a collection of 4-cycles in the Tanner graph. The total number of 4-cycles in the Tanner graph, $\mathcal{N}_4$, is given by
\begin{equation}
\label{counting_cycles_formula_num4cycles_ncxnv_dyadics}
\mathcal{N}_4 = 2^\ell\cdot\sum_{\substack{h,i\in[n_c]\\h<i}}\sum_{\balpha\in\mathbb{F}_2^\ell\cap A_{h,i}}\binom{\rho_{A_{h,i}}(\balpha)}{2},
\end{equation}
where 
$\rho_{A_{h,i}}(\balpha)$ denotes the number of times the element $\balpha$ appears in the mset
$A_{h,i}$.
\end{theorem}
\begin{proof}
By Lemma \ref{counting_cycles_dyadics_imagecycle}, the 4-cycles in the Tanner graph are projected onto TBC walks of length 4 in the protograph and these walks can be studied by computing the product $HH^\mathsf{T}$ as a consequence of Theorem \ref{counting_cycles_dyadics_polynomialpowerofadjacencymatrix}. To determine all possible TBC walk patterns, we use the product $H_2H_2^\mathsf{T}$ instead, but if $H$ has less than two rows, then any pattern arising from these additional row are omitted in the computations of 4-cycles. Once $H_2H_2^\mathsf{T}$ has been computed, we use the strategy of combining walks following Definition \ref{counting_cycles_dyadics_permutationshift} to compute all possible TBC walk patterns, and then use Definition \ref{equivalent_closed_walks} to compute all nonequivalent TBC walk patterns. This gives us the unique TBC walk pattern $[\bh_u,\bi_u,\bi_v,\bh_v]$ if $u\neq v$. If $\bh_u+\bi_u+\bi_v+\bh_v=\bzero$, then $\bh_u+\bi_u=\bh_v+\bi_v$, so we have 4-cycles in the Tanner graph every time this repetition happens. In this case, this repetition contributes $2^\ell$ 4-cycles. Summing over all pair of rows $(h,i)$, $h<i$, gives \eqref{counting_cycles_formula_num4cycles_ncxnv_dyadics}. 
\end{proof}
\begin{rem}
From Theorem~\ref{counting_cycles_thm_4cycles_ncxnv_dyadics} and the preceding discussion, we obtain conditions on avoiding the girth 4 barrier. 
Namely,  we have that $\text{girth}(H)>4$ if and only if each one of the $\binom{n_c}{2}$ 
msets 
$\left\{\bh_u+\bi_u\mid u\in[n_v]\right\}$, for $h<i$, contains distinct elements, which also corresponds to yielding $\cN_4 = 0$ in~\eqref{counting_cycles_formula_num4cycles_ncxnv_dyadics}.
\end{rem}

Algorithm \ref{counting_cycles_algorithm_4cycles_dyadics} shows how we use Theorem \ref{counting_cycles_thm_4cycles_ncxnv_dyadics} to count 4-cycles in the Tanner graph of a quasi-dyadic code with parity-check matrix $H$ as in \eqref{parity_check_matrix_dyadics}. 
\begin{algorithm}[t]
\caption{Counting 4-cycles}
\label{counting_cycles_algorithm_4cycles_dyadics}
\begin{algorithmic}[0]
\Input{Exponents in parity-check matrix \eqref{parity_check_matrix_dyadics}, $n_c$, $n_v$, $\ell$.}
\State{Initialize $\mathcal{N}_4=0$.}
\For{$h=0$ to $n_c-1$}
	\For{$i=h+1$ to $n_c-1$}
	\State $A_{h,i}\gets\{\bh_u+\bi_u\mid u\in[n_v]\}$
		\For{$\balpha\in\mathbb{F}_2^\ell\cap A_{h,i}$}
			\State $\mathcal{N}_4+=\displaystyle\binom{\rho_{A_{h,i}}(\balpha)}{2}$
		\EndFor
	\EndFor
\EndFor
\State \Return $\mathcal{N}_4$
\end{algorithmic}
\end{algorithm}
In Appendix \ref{sec:counting_cycles_dyadics_continuation}, we present the results used to count 6-cycles and 8-cycles in a quasi-dyadic code with parity-check matrix $H$ as in \eqref{parity_check_matrix_dyadics} along with some examples.

We end this section with two observations on the girth.

\begin{theorem}
\label{girth_upper_bound_nadic}
Let $G$ be a protograph with girth $g$. 
Then the girth of any lifting of $G$ using dyadic matrices is upper bounded by $2g$.
\end{theorem}
\begin{proof}
Let $W$ be a $g$-cycle in the protograph, and suppose that $W$ is described by the TBC walk of length $g$ given by $[\bx_{i_1,j_1},\bx_{i_2,j_1},\bx_{i_2,j_2},\bx_{i_3,j_2}\dots,\bx_{i_{g/2},j_{g/2}},\bx_{i_1,j_{g/2}}]$. The permutation shift of this TBC walk is given by $\bs=\bx_{i_1,j_1}+\bx_{i_2,j_1}+\bx_{i_2,j_2}+\bx_{i_3,j_2}+\cdots+\bx_{i_{g/2}+j_{g/2}}+\bx_{i_1,j_{g/2}} \in\F_2^\ell$.
If $\bs$ is nonzero, then it has order 2 in $\F_2^\ell$, that is, $2\bs = {\bf 0}$ and $P_{2\bs} = P_{\bf 0} = I_{2^\ell}$.
Hence, there is a $2g$-cycle in the Tanner graph. 
\end{proof}

\begin{rem}
Analogous to dyadic matrices, one can also consider $n$-adic matrices~\cite{mpk24} where the addition of the subscripts is done over $\Z_n^\ell$ instead, that is, $D_{\bx,\by} = D_{{\bf 0},\bx+\by}$ for all $\bx,\by\in \Z_n^\ell$.
It was shown in~\cite{mpk24} that if a protograph has $\boldone_2$ in its base matrix, then the girth of any quasi-$n$-adic lifting of this protograph is upper-bounded by $4n$. 
Theorem~\ref{girth_upper_bound_nadic} generalizes this result in the sense that if $G$ is a prograph of girth $g$, then the girth of any lifting of $G$ using $n$-adic matrices is upper bounded by $ng$.
The proof of this statement is completely analogous to the dyadic case where one simply considers the order of $\bs$ in $\Z_n^\ell$ instead.
\end{rem}

By Theorem \ref{girth_upper_bound_nadic}, any quasi-dyadic code constructed by lifting a binary base matrix will have girth upper bounded by $2g$, where $g$ is the girth of the protograph. In the case of the all-ones protograph, this upper bound is 8 since the girth of the protograph is 4. To exceed girth 8 in the Tanner graph, we will need the protograph to have girth at least 6.

\subsection{Dyadic-PEG and Cycle Optimization}
\label{sec:counting_cycles_dyadics_cycle_optimization}

The short cycles in the Tanner graph of LDPC codes are known to play an important role in their performance \cite{mb01,hc06}, which motivates the search for LDPC codes from graphs of large girth. However, even if two LDPC codes have the same girth, their performance can be dramatically different if they have a different number of short cycles \cite{hc06}. In an attempt to construct quasi-dyadic LDPC codes with a large girth and a small number of short cycles, we present some construction strategies using a PEG-like approach similar to that presented in \cite{gfsm23b}. 

The Progressive Edge-Growth (PEG) algorithm \cite{hea05} was originally proposed as a general method for constructing Tanner graphs with large girth. This algorithm works by establishing edges between symbol and check nodes in an edge-by-edge manner while satisfying some node degree sequence. The creation of a new edge in the Tanner graph occurs in a position where the impact on the girth is minimized, meaning that edges forming cycles are avoided for as long as possible. In terms of parity-check matrices, the algorithm starts creating a parity-check matrix from an all-zero matrix, placing ones in particular entries that avoid, or delay if unavoidable, the creation of cycles while preserving either a column or row weight distribution. 

The complexity of the PEG algorithm scales as $\mathcal{O}(mn)$, where $m$ is the number of check nodes and $n$ is the number of symbol nodes in the Tanner graph. In a quasi-dyadic LDPC code, these values are given by $m=n_cN$ and $n=n_vN$. Due to the nature of dyadic matrices, once an edge is created between the first check node and some variable node in a given $N\times N$ block, which corresponds to adding a 1 to the signature row of that $N\times N$ dyadic matrix, the remaining $N-1$ edges corresponding to that block will also be created. This allows for a reduction in the complexity of the construction when combined with the girth conditions presented in this section and Appendix \ref{appendix:girth_analysis}. To propose this low-complexity PEG algorithm, we start with a definition.
\begin{definition}[\hspace{-1sp}{\cite{gfsm23b}}]
\label{forbidden_set}
Let the parity-check matrix $H$ of a quasi-dyadic LDPC code be as in \eqref{parity_check_matrix_dyadics} and let $N=2^\ell$ be the lifting factor. The \textbf{forbidden set} $\mathcal{F}_{\ba}^{g,N}$ of the index $\ba$, corresponding to the dyadic permutation matrix $P_{\ba}$ in $H$, is the set of elements in $\mathbb{F}_2^\ell$ such that if $\ba$ assumes any value in this set, then $\text{girth}(H)\leq g$.
\end{definition}

\begin{remark}
Notice that the description of the forbidden set $\mathcal{F}_{\bx}^{g,N}$ of an index $\bx$ corresponding to the dyadic permutation matrix $P_{\bx}$
are given symbolically. In practice and as we will see in a moment, the dyadic permutation matrices are chosen sequentially (one at a time), which means that at a given iteration, some of the indices have not been assigned an element of $\mathbb{F}_2^\ell$ and should not be involved in the computations of that specific forbidden set. As an illustration, 
consider the matrix
\begin{equation*}
H = \matrix{P_{\bh_{0}}&P_{\bh_{1}}&P_{\bh_{2}}&\cdots&P_{\bh_{n_v-1}} \\ P_{\bi_{0}}&P_{\bi_{1}}&P_{\bi_{2}}&\cdots&P_{\bi_{n_v-1}} 
}
\end{equation*}
and suppose that only the dyadic permutation matrices $P_{\bh_0}$, $P_{\bh_1}$, and $P_{\bi_0}$ have been chosen with, say, $\bh_0=\bh_1=\bi_0=\bzero$, and that we want to choose $P_{\bi_1}$. Then we have $\mathcal{F}_{\bi_1}^{4,N}=\left\{\bh_1+\bh_0+\bi_0\right\}=\left\{\bzero+\bzero+\bzero\right\}=\left\{\bzero\right\}$. This means that $\bi_1\neq\bzero$ guarantees girth larger than 4. Notice that the element $\bh_1+\bh_2+\bi_2$, for example, is not an element of $\mathcal{F}_{\bi_1}^{4,N}$ since $\bh_2$ and $\bi_2$ have not been assigned an element of $\mathbb{F}_2^\ell$ before this iteration. 
\end{remark}

A natural way to combine the PEG algorithm with the forbidden set approach a as follows. First, we need to decide in what order we construct the parity-check matrix \eqref{parity_check_matrix_dyadics}, that is, in what order we assign elements in $\mathbb{F}_2^\ell$ to the indices. Two common orderings are row-by-row \[(\bx_{0,0},\bx_{0,1},\dots,\bx_{0,n_v-1},\bx_{1,0},\bx_{1,1},\dots,\bx_{1,n_v-1},\dots,\bx_{n_c-1,0},\bx_{n_c-1,1},\dots,\bx_{n_c-1,n_v-1})\] and column-by-column \[(\bx_{0,0},\bx_{1,0},\dots,\bx_{n_c-1,0},\bx_{0,1},\bx_{1,1},\dots,\bx_{n_c-1,1},\dots,\bx_{0,n_v-1},\bx_{1,n_v-1},\dots,\bx_{n_c-1,n_v-1});\] a random ordering is another option. After choosing the ordering and prescribing a target girth $g$, which is usually taken to be the maximum attainable girth (since PEG tries to guarantee maximum girth), we compute the forbidden set $\mathcal{F}_{\bx}^{g,N}$ of the first index $\bx$ in the ordering and assign to it an element from the set $\mathbb{F}_2^\ell\backslash\mathcal{F}_{\bx}^{g,N}$ at random. For the second index $\by$, we compute $\mathcal{F}_{\by}^{g,N}$ and assign to it an element from the set $\mathbb{F}_2^\ell\backslash\mathcal{F}_{\by}^{g,N}$ at random. The algorithm continues in this fashion for as long as possible. At some point, if $\ell$ is not large enough, then there exists the possibility that the set $\mathbb{F}_2^\ell\backslash\mathcal{F}_{\bz}^{g,N}=\emptyset$ for some index $\bz$. In this case, the algorithm decreases the target girth to $g-2$ and repeats the same strategy, only reducing the target girth when needed. In the quasi-dyadic case, the maximum attainable girth in a lifting of \eqref{parity_check_matrix_dyadics} is 8, so the algorithm will decrease the target girth to either 6 or 4 when needed.

The construction strategy described in the previous paragraph is recorded as Algorithm \ref{dyadic_PEG_girth}.
\begin{algorithm}[t]
\caption{Quasi-dyadic PEG algorithm with forbidden sets}
\label{dyadic_PEG_girth}
\begin{algorithmic}[0]
\Input{$n_c$, $n_v$, $\ell$.}
\State{Initialize indices $\bx_{lk}=\infty$, $l\in[n_c]$, $k\in[n_v]$.}
\For{$k=0$ to $n_v-1$}
	\For{$l=0$ to $n_c-1$}
		\State \multiline{Choose largest $g\in\{4,6\}$ such that $F=\mathbb{F}_2^\ell\backslash\Call{forbidden}{H,\bx_{lk},g,N}\neq\varnothing$}
		\If{there is such $g$}
			\State $\bx_{lk} \gets \Call{random}{F}$	
		\Else
			\State $\bx_{lk} \gets \Call{random}{\mathbb{F}_2^\ell}$	
		\EndIf
	\EndFor
\EndFor
\State \Return $H$
\end{algorithmic}
\end{algorithm}
For a set $S$, the function $\Call{random}{S}$ returns a random element in $S$; the function $\Call{len}{V}$ returns the length of the list $V$ (or the cardinality of the set $V$); and, the function $\Call{forbidden}{H,\texttt{var},g,N}$ returns the set $\mathcal{F}_{\texttt{var}}^{g,N}$. The indices $\bx_{lk}$ are initialized to $\infty$, which can be understood as having a $2^\ell\times2^\ell$ all-zero matrix in the position of what will be the dyadic permutation matrix $P_{\bx_{lk}}$ once we assign to the index $\bx_{lk}$ an element from $\mathbb{F}_2^\ell\backslash\mathcal{F}_{\texttt{var}}^{g,N}$. Notice that Algorithm \ref{dyadic_PEG_girth} constructs an $n_c\times n_v$ parity-check matrix $H$ as in \eqref{parity_check_matrix_dyadics}, following a forward search manner, maximizing the girth where possible from previous index selections. The indices are chosen in a column-by-column fashion, following the idea of PEG. For the row-by-row ordering, we just need to interchange the order of the two \textbf{for} loops in Algorithm \ref{dyadic_PEG_girth}.

Algorithm \ref{dyadic_PEG_girth} can be modified to include a cycle optimization step in the selection stage of an index $\bx$. After choosing the ordering and prescribing a target girth $g$, we start the algorithm by computing the forbidden set $\mathcal{F}_{\bx}^{g,N}$ of the first index $\bx$. Instead of assigning to it an element from the set $\mathbb{F}_2^\ell\backslash\mathcal{F}_{\bx}^{g,N}$ at random, we first compute the number of $g$-cycles, $\mathcal{N}_g$, of the matrix $H$ under the current setting for each element of $\mathbb{F}_2^\ell\backslash\mathcal{F}_{\bx}^{g,N}$. Then, we let $\mathcal{M}\subseteq\mathbb{F}_2^\ell\backslash\mathcal{F}_{\bx}^{g,N}$ be the subset of elements that minimize $\mathcal{N}_g$ under the current setting and we assign to $\bx$ an element from $\mathcal{M}$ at random. This means that in each iteration, we are choosing a dyadic permutation matrix that minimizes the number of cycles introduced by the selection of a new index. To choose the remaining indices $\by$, this procedure is repeated for as long as possible. As before, if at some point it turns out that $\ell$ is not large enough and we have that the set $\mathbb{F}_2^\ell\backslash\mathcal{F}_{\bz}^{g,N}=\emptyset$ for some index $\bz$, the algorithm decreases the target girth to $g-2$ and repeats the same strategy. This algorithm is recorded as Algorithm \ref{dyadic_PEG_girth_minimize_cycles}. 
\begin{algorithm}[t]
\caption{Quasi-dyadic PEG algorithm with forbidden sets and cycle minimization}
\label{dyadic_PEG_girth_minimize_cycles}
\begin{algorithmic}[0]
\Input{$n_c$, $n_v$, $\ell$.}
\State{Initialize indices $\bx_{lk}=\infty$, $l\in[n_c]$, $k\in[n_v]$.}
\For{$k=0$ to $n_v-1$}
	\For{$l=0$ to $n_c-1$}
		\State \multiline{Choose largest $g\in\{4,6\}$ such that $F=\mathbb{F}_2^\ell\backslash\Call{forbidden}{H,\bx_{lk},g,N}\neq\varnothing$}
		\If{there is such $g$}
            \State $\mathcal{M}\gets\left\{\ba\in F\mid\mathcal{N}_g\;\text{is minimum under current setting}\right\}$
			\State $\bx_{lk} \gets \Call{random}{\mathcal{M}}$	
		\Else
            \State $\mathcal{M}\gets\left\{\ba\in \mathbb{F}_2^\ell\mid\mathcal{N}_g\;\text{is minimum under current setting}\right\}$
			\State $\bx_{lk} \gets \Call{random}{\mathbb{F}_2^\ell}$	
		\EndIf
	\EndFor
\EndFor
\State \Return $H$
\end{algorithmic}
\end{algorithm}
Two variations of this algorithm would be to assign to an index $\bx$ an element from $\mathbb{F}_2^\ell\backslash\mathcal{F}_{\bx}^{g,N}$ that maximizes $\mathcal{N}_g$ under the current setting, or an element in $\mathbb{F}_2^\ell\backslash\mathcal{F}_{\bx}^{g,N}$ such that $\mathcal{N}_g$ is average among all the options.


\section{Absorbing Set Analysis}
\label{sec:absorbing_set_analysis}

While short cycles influence decoder convergence, error-floor behavior in iterative decoding is often dominated by more subtle subgraphs. In this section, we turn our attention to absorbing sets, which capture persistent local trapping configurations beyond pure cycle structure. Absorbing sets are combinatorial objects studied in classical coding theory \cite{dolecek2009analysis, D10, dolecek2010towards, beemer2019absorbing, mcmillon2023extremal} and recently analyzed in the quantum setting as well\cite{morris2023analysis, Morris2026AMC}. They have been shown to be the dominant problematic structure affecting the error floor for certain families of LDPC codes, as well as being responsible for many instances of decoding failure under the Gallager B syndrome-based iterative decoder for quantum LDPC codes. Thus, it is a natural question to characterize the presence of absorbing sets for codes arising from dyadic and quasi-dyadic matrices. First, we recall that given a graph $G=(V, E)$ and subset of vertices $W \subseteq V$, the \textbf{induced subgraph} $G_W$ is the graph obtained by taking all vertices in $W$ and all edges incident to these vertices. 

\begin{definition}[\hspace{-1sp}\cite{dolecek2009analysis}]
An \textbf{$(a,b)$-absorbing set $\mathcal{A}$} in a Tanner graph $G$ is a subset of $|\mathcal{A}|=a$ variable nodes such that the induced subgraph $G_{\mathcal{A}\cup \mathcal{N}(\mathcal{A})}$ has $b$ odd degree check nodes and every variable node $v \in \mathcal{A}$ has strictly more even than odd degree neighbors in $G_{\mathcal{A}\cup \mathcal{N}(\mathcal{A})}$.
\end{definition}

We first make some observations about the ``boundary cases'' for dyadic and quasi-dyadic constructions. Observe that Tanner graphs arising from $m \times 1$ arrays of dyadic permutation matrices are acyclic and have no absorbing sets. Indeed, an $m \times 1$ array of dyadic permutation matrices has a Tanner graph $G$ with all check nodes of degree one, which means both that $G$ is acyclic and it has no absorbing sets. Thus, while it is not interesting to consider $m \times 1$ arrays of dyadic permutation matrices with respect to absorbing sets,  $1 \times n$ arrays of dyadic permutation matrices and  their absorbing sets deserve attention. 

Unless stated otherwise, we count absorbing sets whose induced subgraph is connected.

\begin{lemma}
Consider a $1 \times n$ array $M$ of $2^{\ell}\times 2^{\ell}$ dyadic permutation matrices. The corresponding Tanner graph is acyclic and has exactly $2^{\ell}(2^{n-1}-1)$ absorbing sets, all of which are $(a,0)$-absorbing sets for $a=2k$ even and $2 \leq a \leq n$.
\end{lemma} 

\begin{proof}
Let $G$ denote the Tanner graph corresponding to $M$. Since $M$ is a $1 \times n$ array, all variable nodes have degree $1$, so $G$ is acyclic. 
Each check node is incident to exactly one variable node in each block of $M$. 
Thus, each check node is incident to $n$ variable nodes, and the Tanner graph $G$ consists of $2^{\ell}$ disconnected star graphs, meaning complete bipartite graphs $K_{1,t}$, one per check node. 

To count the number of absorbing sets, we count the number of ways to choose an even number of variable nodes from each cluster of $n$ variable nodes and then multiply by $2^{\ell}$ to denote the $2^{\ell}$ clusters forming star graphs. Then the total number of absorbing sets for $G$ is 
\[
2^{\ell}\sum_{k=1}^{\lfloor \frac{n}{2} \rfloor} \binom{n}{2k} = 2^{\ell}(2^{n-1}-1).\\
\]
We note that each cluster has $2^{n-1}-1$ ways of forming absorbing sets since we exclude selecting $0$ variable nodes from the group of $n$ nodes. 
\end{proof}

Next, we consider various configurations of quasi-dyadic whose corresponding Tanner graphs contain complete bipartite graphs. We then present a unified description of their absorbing set structure. 

\begin{lemma}\label{absorbing_identical_blocks}
Consider an $m \times n$ array $M$ of identical dyadic permutation matrices. The Tanner graph $G_M$ comprises  $2^{\ell}$ $K_{m,n}$ complete bipartite graphs, where $m$ and $n$ are the number of check and variable nodes in each $K_{m,n}$ graph, respectively.
\end{lemma}

\begin{proof}
Consider an $m \times n$ array $M$ of identical dyadic permutation matrices. Given a labeling of the rows and columns of $M$, label the variable nodes and check nodes of $G_M$ as 
\begin{equation*}
v_0, \dots, v_{2^{\ell}-1}, v_{2^{\ell}}, \dots, v_{(n-1)2^{\ell}}, \dots, v_{n2^{\ell}-1} \;\text{and}\; c_0, \dots, c_{2^{\ell}-1}, c_{2^{\ell}}, \dots, c_{(n-1)2^{\ell}}, \dots, c_{n2^{\ell}-1}.
\end{equation*}
Since $M$ comprises dyadic permutation matrices, the degree of each variable node and each check node is $m$ and $n$, respectively. Consider some variable node $v_i$ for $0 \leq i \leq 2^{\ell}-1$. That is, $v_i$ corresponds to a column in the first $m \times 1$ column of the $m \times n$ array. We know that $v_i \sim c_{i_0}$ for $0 \leq i_0 \leq 2^{\ell}-1$. Moreover, $v_{i+k\cdot 2^{\ell}} \sim c_{i_0}$ for $0 \leq k \leq n-1$. Since $\text{deg}(c_{i_0})=n$, the neighborhood $\mathcal{N}(c_{i_0})$ is precisely $\{v_{i+k\cdot 2^{\ell}} \: | \: 0 \leq k \leq n-1 \}$. 
On the other hand, $v_i \sim c_{i_j}$ for $0 \leq j \leq m-1$. Thus, $\mathcal{N}(c_{i_j}) = \{v_{i+k\cdot 2^{\ell}} \: | \: 0 \leq k \leq n-1 \}$ for $0 \leq j \leq m-1$ and $\mathcal{N}(v_{i+k2^{\ell}})=\{c_{i_j} \: | \: 0 \leq j \leq m-1\}$ for $0 \leq k \leq n-1$. 

By degree counting, the induced graph $G_i = (V_i, C_i, E_i)$ formed by $V_i = \{ v_{i+k\cdot 2^{\ell}} \: | \: 0 \leq k \leq n-1\}$ and $C_i = \{c_{i_j} \: | \: 0 \leq j \leq m-1 \}$ is the complete bipartite graph $K_{m,n}$ on $m$ check nodes and $n$ variable nodes. 
Since we considered an arbitrary index $i$ such that $0 \leq i \leq 2^{\ell}-1$, $G_M$ thus comprises $2^{\ell}$ $K_{m,n}$ complete bipartite graphs.
\end{proof}

\begin{lemma}\label{block_diagonal_structure}
Consider a dyadic matrix $D\in \mathbb{F}_2^{2^{\ell}\times 2^{\ell}}$ with signature consisting of $2^k$ ones at the beginning followed by zeros.
Then $G_{D}$ is a disconnected graph, consisting of $2^{\ell - k}$ subgraphs $K_{2^k,2^k}$. 
\end{lemma}

\begin{proof} 
Given that the signature of $D$ is of form $\sigma(D) = ({\bf 1}_{2^k},{\bf 0}_{2^{\ell-k}})$, in the language of Section 2-A, we have that $\sigma(M)$ is supported on a subspace of dimension $2^k$. 
As such, $D = {\sf diag}(\boldone_{2^k},\cdots, \boldone_{2^k})$; see also Theorem~\ref{T-space}.

The block diagonal nature of $D$ immediately yields that the Tanner graph $G_D$ has $2^{\ell-k}$ connected components, and each block $D_{i,i}$ for $0 \leq i \leq 2^{\ell-k}$ being an all-ones matrix tells us these connected components are all $K_{2^k, 2^k}$ graphs.
\end{proof} 

Since the graph structure consists of multiple copies of $K_{m,n}$ complete bipartite graphs, we consider a single dyadic whose Tanner graph is complete and then count the resulting number of absorbing sets in a single complete bipartite graph. 

\begin{lemma}\label{absorbing_complete_graphs}
Let $D = \boldone_{2^\ell}$ be the all-one dyadic matrix.
Then, all absorbing sets of the corresponding Tanner graph are of the form $(a,0)$ for $a$ even and $2 \leq a \leq 2^{\ell}$. 
Moreover, there are $\binom{2^{\ell}}{a}$ such $(a,0)$ absorbing sets. Setting $a:=2k$ for $1 \leq k \leq 2^{\ell-1}$ there are $\sum_{k=1}^{2^{\ell-1}}\binom{2^{\ell}}{2k}$ total absorbing sets in the corresponding Tanner graph.
\end{lemma}

\begin{proof}
In this case it is clear that $G_D$ is the complete bipartite graph on $2 \cdot 2^{\ell}$ vertices. Thus, the induced graph of any subset $\mathcal{A}$ consisting of $a$ variable nodes will form a complete bipartite graph on $2a$ nodes, implying that each variable node in $G_{\mathcal{A}}$ will have $a$ degree $a$ check neighbors. Hence, absorbing sets consist of all even-cardinality sets of variable nodes, forming $(a,0)$ absorbing sets. For a given value $a$, there are $\binom{2^{\ell}}{a}$ ways of obtaining an $(a,0)$ absorbing set. Thus, the total number of absorbing sets arises from summing over all nonzero even values of $a \leq 2^{\ell}$.  
\end{proof}

We use Lemma \ref{absorbing_complete_graphs} and combine the absorbing set results for Lemma \ref{absorbing_identical_blocks} and Lemma \ref{block_diagonal_structure} in the following corollary. 

\begin{corollary} The Tanner graphs for dyadic matrices as in Lemma \ref{absorbing_identical_blocks} and Lemma \ref{block_diagonal_structure} have absorbing sets of the following type: 
\begin{arabiclist}
\item For Lemma~\ref{absorbing_identical_blocks}, all absorbing sets whose induced subgraphs are connected are of the form $(a,0)$ for $a$ even and $2 \leq a \leq \lfloor \frac{n}{2} \rfloor \cdot 2$. There are $n \cdot \binom{n}{a}$ such $(a,0)$ absorbing sets.
\item For Lemma \ref{block_diagonal_structure}, all absorbing sets whose induced subgraphs are connected are of the form $(a,0)$ for $a$ even and $2 \leq a \leq 2^{k}$. There are $2^{\ell-k} \cdot \binom{2^k}{a}$ such $(a,0)$ absorbing sets.
\end{arabiclist}
\end{corollary}

\begin{proof}
For any graph consisting of connected components, absorbing sets whose induced subgraphs are connected will thus consist of variable nodes from a single connected component. For Lemma \ref{absorbing_identical_blocks}, whose connected components are $K_{m,n}$ complete bipartite graphs, even if $n$ is odd, absorbing sets whose induced subgraphs are connected will again consist of all even collection of variable nodes from a single $K_{m,n}$ complete bipartite graph. 
Then, since there are $n$ complete bipartite graphs $K_{m,n}$, there are $n \cdot \binom{n}{a}$ such $(a,0)$ absorbing sets, for $2 \leq a \leq \lfloor \frac{n}{2} \rfloor \cdot 2$.

The graph $G_D$ consists of $2^{\ell - k}$ disconnected $K_{2^k,2^k}$ subgraphs. Thus, absorbing sets whose induced subgraphs are connected arise from one of these $K_{2^k, 2^k}$ subgraphs. By Lemma \ref{absorbing_complete_graphs}, all of these absorbing sets are of the form $(a,0)$ for $a$ even and $2 \leq a \leq 2^{\ell}$. This means we have $2^{\ell-k}$ groups of $2^k$ variable nodes from which we choose $a$ variable nodes to form an absorbing set. This number is $2^{\ell-k}\binom{2^k}{a}.$
\end{proof}


\section{Simulation Results}
\label{sec:simulation_results}

In this section, we provide simulations to demonstrate the performance of the codes considered in this paper. For simulations, we use the belief propagation decoding algorithm~\cite{BP} using the software~\cite{Roffe_LDPC_Python_tools_2022}.

\begin{figure}[htb!]
\centering
\subfloat[][]{
\includegraphics[width=0.3\paperwidth]{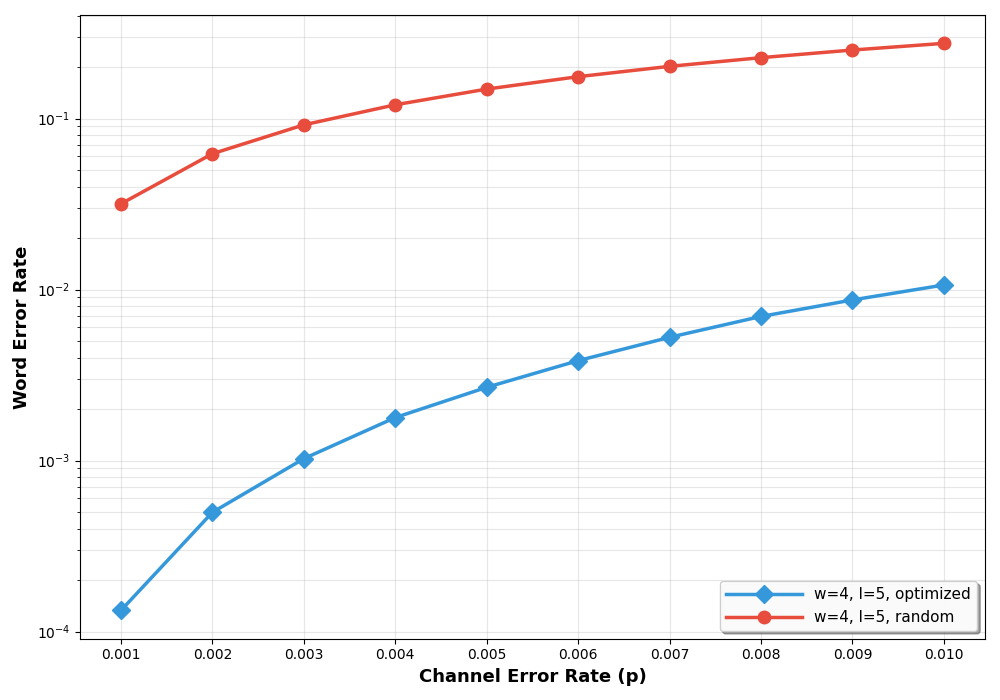}
\label{fig:subfig1}}
\subfloat[][]{
\includegraphics[width=0.3\paperwidth]{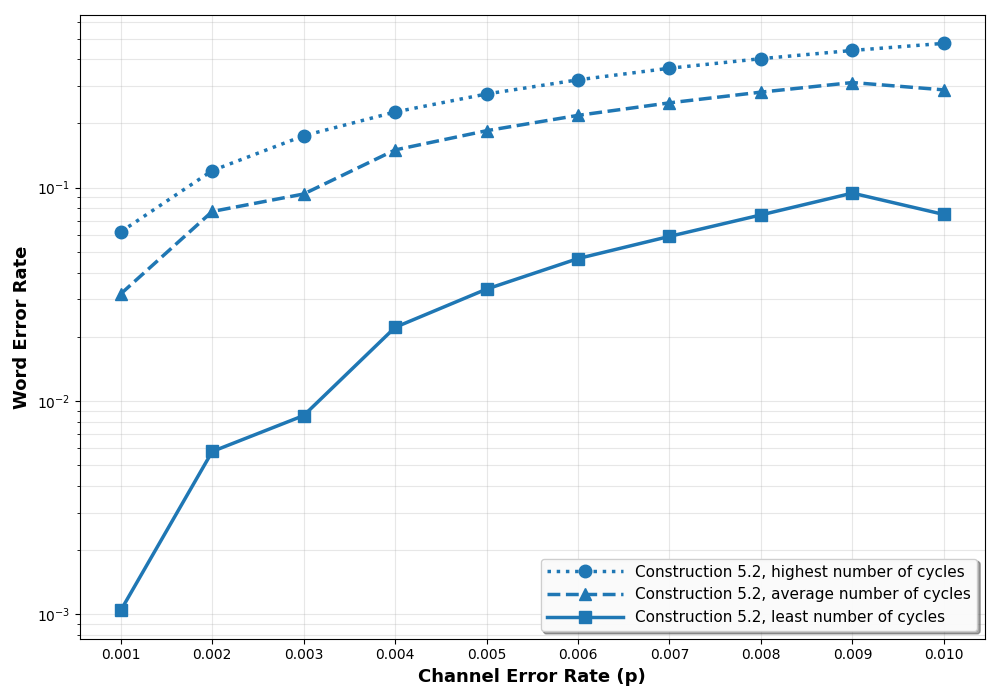}\label{fig:subfig2}}\\
\subfloat[][]{
\includegraphics[width=0.3\paperwidth]{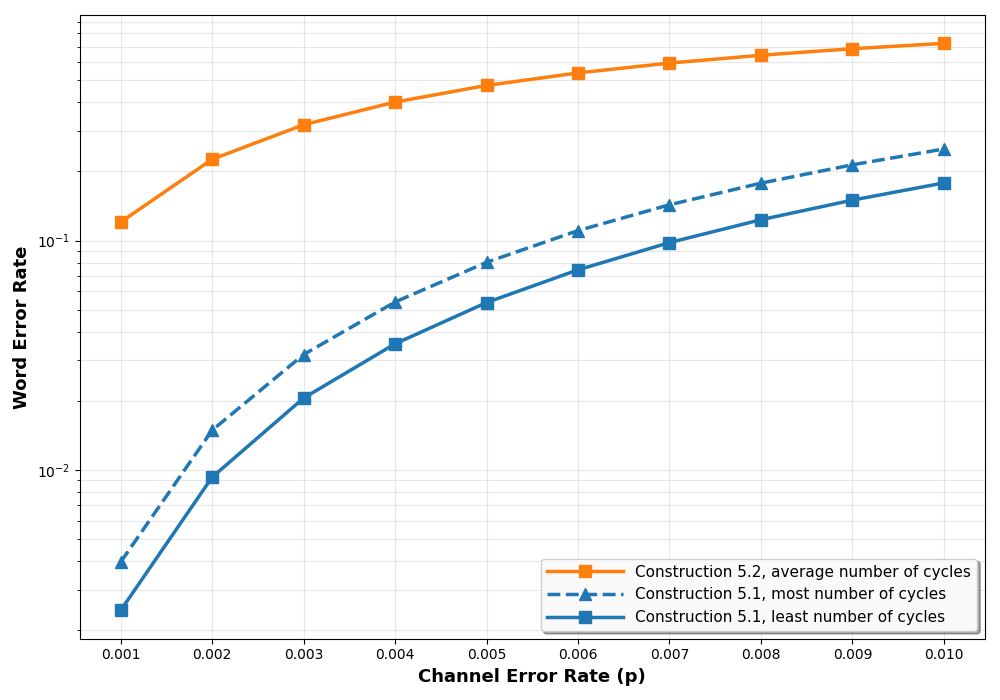}\label{fig:subfig3}}
\caption{Subfigure (a): Two dyadic matrices with signature weight 4 and $\ell = 5$. Subfigure (b): Construction~\ref{con-BBS} with $\ell=3$ and $v = 4$ using only the first three rows. Subfigure (c): Comparison of Construction~\ref{con-main} versus Construction~\ref{con-BBS}.}
\label{Fig1}
\end{figure}

As a proof of concept, we begin by showcasing our cycle analysis with two dyadic matrices. 
In Figure~\ref{Fig1}(a), we use dyadic matrices with signature weight 4 and $\ell = 5$, where one matrix is random whereas the other is optimized so that the resulting Tanner graph has minimal number of 4-cycles.
The decoding gain in this case is clear.
This behavior is caused by dyadic matrices that are supported on subspaces or cosets and their ``duals'' (where a signature $\bv$ becomes $\bv + {\bf 1}$) all of which have 288 4-cycles. 
In comparison, our ``optimized" matrix has only 96 4-cycles.

Next, in Figure~\ref{Fig1}(b), we test our cycle analysis on Construction~\ref{con-BBS}. 
Here we use $\ell=3$ and $3\times 8$ block matrices. As discussed, in all scenarios, the girth is four but they differ by the number of 4-cycles. 
The dotted line represents a choice with a maximum number of 4-cycles, the dashed line represents a choice with an average number of 4-cycles (for which we use a representation of~\cite{BBS26}), and for the solid line we use our PEG algorithm to obtain a code with the minimal number of 4-cycles.
As seen in Figure~\ref{Fig1}(b), the optimized version performs significantly better whereas the ``average'' version is closer to the worse version.

In Figure~\ref{Fig1}(c), we use $\ell = 4$ and $2\times 8$ block matrices. The orange line represents Construction~\ref{con-BBS} in its random/average version (as discussed above). 
Whereas the blue lines represent Construction~\ref{con-main} with the dashed line representing the random/average version and the solid line representing the best version (that is, the smallest number of 8-cycles).
Immediately, Construction~\ref{con-main} is expected to perform better since it allows for bigger girth. Then minimizing the number of small cycles yields further improvements.

\section{Conclusions and Future Research}
\label{sec:conclusions}

This paper demonstrates how the dyadic and quasi-dyadic structure can be leveraged as a practical design framework for (Q)LDPC codes by enabling efficient cycle enumeration, dyadic-aware PEG/forbidden-set constructions that either increase girth when possible or, when dyadic constraints limit girth growth, reduce the multiplicity and organization of the shortest cycles, and by providing initial characterizations of absorbing-set behavior for key structured layouts.  

We have shown how systematic control of short-cycle multiplicity and bad layouts can yield tangible iterative-decoding gains. It remains open to establish the fundamental limits of dyadic lifts, including tight bounds on achievable girth and minimal short-cycle multiplicities along with structure–randomness trade-offs and extending absorbing-set results to general quasi-dyadic ensembles and identifying dyadic invariants that predict the error floor. It would also be interesting to explore 
strengthening the constructions and to determine whether forbidden sets are complete for lengths 6 and 8 pathologies along with what optimization objectives best correlate with performance (e.g., cycle counts). 
Along the same lines, in order to break the girth $g \leq 8$ bound for dyadic lifting of the all-one protograph, we will explore in future work $n$-adic lifting of {\em irregular} protographs.
Preliminary results already show that high girths are indeed achievable and a similar cycle analysis is feasible.

Another direction worth pursuit involves addressing quantum-specific issues such as distance and rate scaling for dyadic CSS families along with consideration of how dyadic regularity may improve syndrome extraction and decoder parallelism and generalizations to other group-ring lift families beyond characteristic $2$.


\appendix

\section{Girth Analysis (Continuation)}
\label{appendix:girth_analysis}

In this appendix, we continue the cycle analysis of dyadic and quasi-dyadic codes presented in Sections \ref{sec:counting_cycles_single_dyadics} and \ref{sec:counting_cycles_dyadics}, respectively.

\subsection{Cycles in Dyadic Codes (Continuation)}
\label{sec:counting_cycles_single_dyadics_continuation}

\subsubsection{Counting 6-cycles}
\label{sec:counting_cycles_counting_6cycles_single_dyadics}

If $D$ is a dyadic matrix as in \eqref{parity_check_matrix_single_dyadics} and its signature row $\sigma(D)$ has weight $\omega\geq2$, then $D$ has girth 4 as we discussed in Section \ref{sec:counting_cycles_counting_4cycles_single_dyadics}. We can use our cycle analysis to also compute the number of 6-cycles. By Lemma \ref{counting_cycles_dyadics_imagecycle}, a 6-cycle in the Tanner graph of a dyadic matrix projects only onto a TBC walk of length 6. The reason why a 6-cycle does not project onto a TBC walk of length 2, which would be traversed three times, is because there are no elements of order 3 in $\mathbb{F}_2^\ell$. Also, the reason why it does not project onto a TBC walk of length 3, which would be traversed twice, is because the protograph is bipartite, so there are no 3-cycles in it. These walks can be studied by computing the product $DD^\mathsf{T}D$ according to Theorem \ref{counting_cycles_dyadics_polynomialpowerofadjacencymatrix}. A TBC walk of length 6 in the protograph has the form $[\bx_{u},\bx_{v},\bx_{w},\bx_{w'},\bx_{v'},\bx_{u'}]$ and counting the number of 6-cycles requires analyzing the mset $A=\left\{\left(u,v,w,\bx_{u}+\bx_{v}+\bx_{w}\right)\mid u,v,w\in[\omega],u\neq v,v\neq w\right\}$. Consider the following theorem.
\begin{theorem}
\label{counting_cycles_thm_6cycles_ncxnv_single_dyadics}
Let $D$ be as in \eqref{parity_check_matrix_single_dyadics} and let
\begin{equation*}
A=\left\{\left(u,v,w,\bx_{u}+\bx_{v}+\bx_{w}\right)\mid u,v,w\in[\omega],u\neq v,v\neq w\right\}.
\end{equation*}
\noindent For $u,v,w,u',v',w'\in[\omega]$, let $(u,v,w,\balpha_{u,v,w})$ and $(u',v',w',\balpha_{u',v',w'}) \in A$ be such that $\balpha_{u,v,w}=\balpha_{u',v',w'}$.  
Then this repetition $\balpha_{u,v,w}=\balpha_{u',v',w'}$ lifts to a collection of 6-cycles in the Tanner graph if $u\neq u'$ and $w\neq w'$. The total number of 6-cycles in the Tanner graph, $\mathcal{N}_6$, is given by
\begin{equation}
\label{counting_cycles_formula_num6cycles_ncxnv_single_dyadics}
\mathcal{N}_6=\frac{2^\ell}{6}\cdot\left(\mathcal{R}_{A}^{*c_1}+\mathcal{R}_{A}^{*c_2}+\mathcal{R}_{A}^{*c_3}+\mathcal{R}_{A}^{*c_4}+\mathcal{R}_{A}^{*nc}\right),
\end{equation}
where $\mathcal{R}_{A}^{*c_1}$, $\mathcal{R}_{A}^{*c_2}$, $\mathcal{R}_{A}^{*c_3}$, $\mathcal{R}_{A}^{*c_4}$, and $\mathcal{R}_{A}^{*nc}$ are the numbers of repetitions $\balpha_{u,v,w}=\balpha_{u',v',w'}$ in the mset $A$ satisfying Conditions 1, 2, 3, 4, and not satisfying either, respectively, and where
\begin{itemize}
\item Condition 1: $u=w'$, $v=v'$, and $w=u'$;
\item Condition 2: $v=u'$ and $w=v'$;
\item Condition 3: $u=w$ and $w'=u'$; and
\item Condition 4: $u=v'$ and $v=w'$.
\end{itemize} 
\end{theorem}
\begin{proof}
For $u,v,w,u',v',w'\in[\omega]$, let $(u,v,w,\balpha_{u,v,w})$ and $(u',v',w',\balpha_{u',v',w'})$ be two elements in $A$ such that $\balpha_{u,v,w}=\balpha_{u',v',w'}$. The walk of length 6 corresponding to this repetition is given by $W=[\bx_{u},\bx_{v},\bx_{w},\bx_{w'},\bx_{v'},\bx_{u'}]$. If $u\neq u'$ and $w\neq w'$, then this walk is actually a TBC walk with permutation shift $\balpha_{u,v,w}+\balpha_{u',v',w'}=\bzero$, so by Corollary \ref{counting_cycles_dyadics_preimageTBCwalk_are_cycles}, its inverse image is the union of $N/m$ 6-cycles, where $m\geq 1$ is the order of the permutation shift of $W$ in $\mathbb{F}_2^\ell$. Notice that $m\in\{1,2\}$, but since $m=2$ is not possible by the argument before the statement of Theorem \ref{counting_cycles_thm_6cycles_ncxnv_single_dyadics}, we only study the case $m=1$. 

Let us focus on the indices $(u,v,w,w',v',u')$. When computing all the equivalent walks to $W$, there are five sets of conditions that arise:
\begin{itemize}
\item Condition 1: $u=w'$, $v=v'$, and $w=u'$;
\item Condition 2: $v=u'$ and $w=v'$;
\item Condition 3: $u=w$ and $w'=u'$;
\item Condition 4: $u=v'$ and $v=w'$; and
\item Condition 5: $u=w=v'$ and $v=w'=u'$.
\end{itemize} 
If Condition 1 is assumed, the tuple of indices $(u,v,w,w',v',u')$ becomes $(u,v,w,u,v,w)$. Since the walk $[\bx_{u},\bx_{v},\bx_{w}]$ is not closed, the walk corresponding to the tuple $(u,v,w,u,v,w)$, $W$, is not the double traversal of a TBC walk of length 3 (which would require $m=2$), so the tuple actually contributes to the number of 6-cycles in the Tanner graph. When computing all the walks equivalent to $W$, we obtain exactly 6 distinct elements $(u,v,w,u,v,w)$, $(v,w,u,v,w,u)$, $(w,u,v,w,u,v)$, $(u,w,v,u,w,v)$, $(v,u,w,v,u,w)$, and $(w,v,u,w,v,u)$, which means that the repetition $\balpha_{u,v,w}=\balpha_{u',v',w'}$ is counted six times in the construction of $A$. The number of 6-cycles arising from this repetition is $2^\ell$.

If Condition 2 is assumed, the tuple of indices $(u,v,w,w',v',u')$ becomes $(u,v,w,w',w,v)$.  
When computing all the walks equivalent to $W$, we obtain 6 distinct elements $(u,v,w,w',w,v)$, $(v,w,w',w,v,u)$, $(w,w',w,v,u,v)$, $(w',w,v,u,v,w)$, $(w,v,u,v,w,w')$, and $(v,u,v,w,w',w)$, which means that the repetition $\balpha_{u,v,w}=\balpha_{u',v',w'}$ is counted six times in the construction of $A$. The number of 6-cycles arising from this repetition is $2^\ell$.

If Condition 3 is assumed, the tuple of indices $(u,v,w,w',v',u')$ becomes $(u,v,u,u',v',u')$. 
When computing all the walks equivalent to $W$, we obtain 6 distinct elements $(u,v,u,u',v',u')$, $(v,u,u',v',u',u)$, $(u,u',v',u',u,v)$, $(u',v',u',u,v,u)$, $(v',u',u,v,u,u')$, and $(u',u,v,u,u',v')$, which means that the repetition $\balpha_{u,v,w}=\balpha_{u',v',w'}$ is counted six times in the construction of $A$. The number of 6-cycles arising from this repetition is $2^\ell$.

If Condition 4 is assumed, the tuple of indices $(u,v,w,w',v',u')$ becomes $(u,v,w,v,u,u')$. 
When computing all the walks equivalent to $W$, we obtain 6 distinct elements $(u,v,w,v,u,u')$, $(v,w,v,u,u',u)$, $(w,v,u,u',u,v)$, $(v,u,u',u,v,w)$, $(u,u',u,v,w,v)$, and $(u',u,v,w,v,u)$, which means that the repetition $\balpha_{u,v,w}=\balpha_{u',v',w'}$ is counted six times in the construction of $A$. The number of 6-cycles arising from this repetition is $2^\ell$.

If Condition 5 is assumed, the tuple of indices $(u,v,w,w',v',u')$ becomes $(u,v,u,v,u,v)$. The corresponding TBC walk would be formed by traversing a TBC walk of length 2 three times, but we have argued why this is not possible. Finally, if none of these conditions are satisfied, then we also obtain 6 distinct elements $(u,v,w,w',v',u')$, $(v,w,w',v',u',u)$, $(w,w',v',u',u,v)$, $(w',v',u',u,v,w)$, $(v',u',u,v,w,w')$, and $(u',u,v,w,w',v')$, when computing all the walks equivalent to $W$. This means that the repetition $\balpha_{u,v,w}=\balpha_{u',v',w'}$ is counted six times in the construction of $A$. The number of 6-cycles arising from this repetition is $2^\ell$.

Therefore, if $\mathcal{R}_{A}^{*c_1}$, $\mathcal{R}_{A}^{*c_2}$, $\mathcal{R}_{A}^{*c_3}$, $\mathcal{R}_{A}^{*c_4}$, and $\mathcal{R}_{A}^{*nc}$ are the numbers of repetitions $\balpha_{u,v,w}=\balpha_{u',v',w'}$ in the mset $A$ satisfying Conditions 1, 2, 3, 4, and not satisfying either, respectively, then \eqref{counting_cycles_formula_num6cycles_ncxnv_single_dyadics} follows and we conclude the proof.
\end{proof}

\begin{example}
\label{counting_cycles_example_6cycles_g4_single_dyadics}
Let $D$ be the parity-check matrix given by $D=P_{\mathbf{0}}+P_{\mathbf{1}}+P_{\mathbf{2}}+P_{\mathbf{3}}$ and let $N=2^3$ as in Example \ref{counting_cycles_example_4cycles_g4_single_dyadics}. We know that $D$ has girth 4, so we can use our strategy to count the number of 6-cycles in the Tanner graph. To do this, we construct the mset in Theorem \ref{counting_cycles_thm_6cycles_ncxnv_single_dyadics} and get 
\begin{align*}
A=\{&(0,1,0,\mathbf{1}),(0,1,2,\mathbf{3}),(0,1,3,\mathbf{2}),(0,2,0,\mathbf{2}),(0,2,1,\mathbf{3}),(0,2,3,\mathbf{1}),(0,3,0,\mathbf{3}),(0,3,1,\mathbf{2}), \\ &(0,3,2,\mathbf{1}),(1,0,1,\mathbf{0}),(1,0,2,\mathbf{3}),(1,0,3,\mathbf{2}),(1,2,0,\mathbf{3}),(1,2,1,\mathbf{2}),(1,2,3,\mathbf{0}),(1,3,0,\mathbf{2}), \\ &(1,3,1,\mathbf{3}),(1,3,2,\mathbf{0}),(2,0,1,\mathbf{3}),(2,0,2,\mathbf{0}),(2,0,3,\mathbf{1}),(2,1,0,\mathbf{3}),(2,1,2,\mathbf{1}),(2,1,3,\mathbf{0}), \\ &(2,3,0,\mathbf{1}),(2,3,1,\mathbf{0}),(2,3,2,\mathbf{3}),(3,0,1,\mathbf{2}),(3,0,2,\mathbf{1}),(3,0,3,\mathbf{0}),(3,1,0,\mathbf{2}),(3,1,2,\mathbf{0}), \\ &(3,1,3,\mathbf{1}),(3,2,0,\mathbf{1}),(3,2,1,\mathbf{0}),(3,2,3,\mathbf{2})\}.
\end{align*} 
Some computations show that $\mathcal{R}_{A}^{*c_1}=24$, $\mathcal{R}_{A}^{*c_2}=24$, $\mathcal{R}_{A}^{*c_3}=24$, $\mathcal{R}_{A}^{*c_4}=24$, and $\mathcal{R}_{A}^{*nc}=48$. Using \eqref{counting_cycles_formula_num6cycles_ncxnv_single_dyadics},
\begin{align*}
\mathcal{N}_6 &= \frac{2^\ell}{6}\cdot\left(\mathcal{R}_{A}^{*c_1}+\mathcal{R}_{A}^{*c_2}+\mathcal{R}_{A}^{*c_3}+\mathcal{R}_{A}^{*c_4}+\mathcal{R}_{A}^{*nc}\right) = \frac{2^3}{6}\cdot\left(24+24+24+24+48\right) \\
&= 768,
\end{align*}
so the total number of 6-cycles in the Tanner graph of $D$ is $\mathcal{N}_6=768$.
\end{example}

\begin{remark}
\label{remark_N6_lower_bound}
In the same line of Remark \ref{remark_N4_lower_bound}, if $\omega\geq3$, then notice that $W=[\bx_0,\bx_1,\bx_2,\bx_0,\bx_1,\bx_2]$ is a TBC walk of length 6 that always has permutation shift $\bzero$ in $\mathbb{F}_2^\ell$. This is a TBC walk that satisfies Condition 1 in Theorem \ref{counting_cycles_thm_6cycles_ncxnv_single_dyadics}. By the proof of Theorem \ref{counting_cycles_thm_6cycles_ncxnv_single_dyadics}, notice that $\mathcal{N}_6\geq\frac{2^\ell}{2}\cdot\frac{1}{6}\cdot\omega(\omega-1)(\omega-2)$. The first term of \eqref{counting_cycles_formula_num6cycles_ncxnv_single_dyadics} is equal to the right-hand side of this inequality since that is the number of ways we can construct the TBC walks of length 6 satisfying this condition. If $\omega\in\{0,1,2\}$, we naturally have $\mathcal{N}_6=0$. 
\end{remark}

\subsection{Cycles in Protograph-based LDPC Codes from Dyadics (Continuation)}
\label{sec:counting_cycles_dyadics_continuation}

\subsubsection{Counting 6-cycles}
\label{counting_cycles_counting_6cycles_dyadics}

The 6-cycles in the Tanner graph of protograph-based LDPC codes are projected onto a $3\times3$ submatrix of the form
\begin{equation*}
H_3=\matrix{P_{\bh_u}&P_{\bh_v}&P_{\bh_w} \\ P_{\bi_u}&P_{\bi_v}&P_{\bi_w} \\ P_{\bj_u}&P_{\bj_v}&P_{\bj_w}}.
\end{equation*}
The corresponding TBC walks of length 6 in the protograph obtained from this submatrix are given by 
$[\bh_u,\bi_u,\bi_m,\bj_m,\bj_{u'},\bh_{u'}]$, $[\bi_u,\bh_u,\bh_m,\bj_m,\bj_{u'},\bi_{u'}]$, and $[\bj_u,\bh_u,\bh_m,\bi_m,\bi_{u'},\bj_{u'}]$ for 
$0\leq h<i<j\leq n_c-1$ and $u,m,u'\in[n_v]$. In \cite{gfsm23c}, it was shown that these three TBC walk patterns are equivalent, so only one of them is enough to describe all 6-cycles in the Tanner graph. The first TBC walk pattern gives the product
\begin{align*}
{\left(P_{\bh_{u'}}\right)}^{-1}\cdot P_{\bj_{u'}}\cdot{\left(P_{\bj_m}\right)}^{-1}\cdot P_{\bi_m}\cdot{\left(P_{\bi_u}\right)}^{-1}\cdot P_{\bh_u} &= P_{\bh_{u'}} P_{\bj_{u'}} P_{\bj_m} P_{\bi_m} P_{\bi_u} P_{\bh_u} \\
&= P_{\bh_u+\bi_u+\bi_m+\bj_m+\bj_{u'}+\bh_{u'}},
\end{align*}
so we require that 
$\bh_u+\bi_u+\bi_m+\bj_m+\bj_{u'}+\bh_{u'}\neq\bzero$ to avoid 6-cycles in the Tanner graph. From this observation, we have that $\text{girth}(H)>6$ if and only if for each %
$0\leq h<i<j\leq n_c-1$, all the elements in each one of the msets
\begin{equation*}
\left\{\bh_u+\bi_u+\bi_m,\bh_u+\bj_u+\bj_m \mid u\in[n_v]\backslash\{m\} \right\}, \quad m\in[n_v],
\end{equation*}
contains distinct elements. To count 6-cycles in the Tanner graph, we rewrite these conditions. Consider the following theorem.

\begin{theorem}
\label{counting_cycles_thm_6cycles_ncxnv_dyadics}
Let $H$ be as in \eqref{parity_check_matrix_dyadics}. For 
$0\leq h<i<j\leq n_c-1$ and $m\in[n_v]$, consider the following msets
\begin{align*}
A_{1,m}^{h,i,j} &= \left\{(u,\bh_u+\bi_u+\bi_m) \mid u\in[n_v]\backslash\{m\} \right\}, \\
A_{2,m}^{h,i,j} &= \left\{(u,\bh_u+\bj_u+\bj_m) \mid u\in[n_v]\backslash\{m\} \right\}.
\end{align*}
\noindent For 
$u,u'\in[n_v]$, let $(u,\balpha_u)\in A_{1,m}^{h,i,j}$ and $(u',\balpha_{u'})\in A_{2,m}^{h,i,j}$ be such that $\balpha_u=\balpha_{u'}$. Then the repetition  
$\balpha_u=\balpha_{u'}$ lifts to a collection of 6-cycles in the Tanner graph if 
$u\neq u'$. The total number of 6-cycles in the Tanner graph, $\mathcal{N}_6$, is given by
\begin{equation}
\label{counting_cycles_formula_num6cycles_ncxnv_dyadics}
\mathcal{N}_6 = 2^\ell\cdot\sum_{\substack{h,i,j\in[n_c]\\h<i<j}}\sum_{m\in[n_v]}\sum_{\balpha\in\mathbb{F}_2^\ell}\rho_{A_{1,m}^{h,i,j},A_{2,m}^{h,i,j}}(\balpha),
\end{equation}
where
$\rho_{A_{1,m}^{h,i,j},A_{2,m}^{h,i,j}}(\balpha)$ denotes the number of times the element $\balpha$ is repeated between the msets 
$A_{1,m}^{h,i,j}$ and $A_{2,m}^{h,i,j}$ in the sense described above.
\end{theorem}
\begin{proof}
By Lemma \ref{counting_cycles_dyadics_imagecycle}, the 6-cycles in the Tanner graph are projected onto TBC walks of length 6 in the protograph and these walks can be studied by computing the product $HH^\mathsf{T}H$ as a consequence of Theorem \ref{counting_cycles_dyadics_polynomialpowerofadjacencymatrix}. To determine all possible TBC walk patterns, we use the product $H_3H_3^\mathsf{T}H_3$ instead, but if $H$ has less than three rows, then any pattern arising from these additional rows are omitted in the computations of 6-cycles. Once $H_3H_3^\mathsf{T}H_3$ has been computed, we use the strategy of combining walks following Definition \ref{counting_cycles_dyadics_permutationshift} to compute all possible TBC walk patterns, and then use Definition \ref{equivalent_closed_walks} to compute all nonequivalent TBC walk patterns. This gives us the unique TBC walk pattern $[\bh_u,\bi_u,\bi_m,\bj_m,\bj_{u'},\bh_{u'}]$ if $u\neq m$, $m\neq u'$, and $u'\neq u$. If $\bh_u+\bi_u+\bi_m+\bj_m+\bj_{u'}+\bh_{u'}=\bzero$, then $\bh_u+\bi_u+\bi_m=\bh_{u'}+\bj_{u'}+\bj_m$, so we have 6-cycles in the Tanner graph every time this repetition happens. In this case, this repetition contributes $2^\ell$ 6-cycles. If we sum over all sets of three of rows $(h,i,j)$, $h<i<j$, and all $m\in[n_v]$, \eqref{counting_cycles_formula_num6cycles_ncxnv_dyadics} follows.
\end{proof}

\begin{example}
\label{counting_cycles_example_6cycles_g6_dyadics}
Let $H$ be the parity-check matrix given by 
\begin{equation*}
H = \matrix{ P_{\mathbf{0}}&P_{\mathbf{0}}&P_{\mathbf{0}}&P_{\mathbf{0}}&P_{\mathbf{0}} \\ P_{\mathbf{0}}&P_{\mathbf{1}}&P_{\mathbf{2}}&P_{\mathbf{3}}&P_{\mathbf{4}} \\ P_{\mathbf{0}}&P_{\mathbf{2}}&P_{\mathbf{4}}&P_{\mathbf{6}}&P_{\mathbf{8}}}.
 \end{equation*}
For $N=2^4$, this parity-check matrix has girth 6. To show this, we construct the msets in Theorem \ref{counting_cycles_thm_6cycles_ncxnv_dyadics} and check for repetitions. Since $H$ has exactly three rows, the only combination of three rows is given by 
$(h,i,j)=(0,1,2)$, so the only msets that we need to construct are the following:
\begin{multicols}{2}
\noindent
\begin{align*}
A_{1,0}^{0,1,2} &= \{(1, \mathbf{1}), (2, \mathbf{2}), (3, \mathbf{3}), (4, \mathbf{4})\}, \\
A_{2,0}^{0,1,2} &= \{(1, \mathbf{2}), (2, \mathbf{4}), (3, \mathbf{6}), (4, \mathbf{8})\}, \\
A_{1,1}^{0,1,2} &= \{(0, \mathbf{1}), (2, \mathbf{3}), (3, \mathbf{2}), (4, \mathbf{5})\}, \\
A_{2,1}^{0,1,2} &= \{(0, \mathbf{2}), (2, \mathbf{6}), (3, \mathbf{4}), (4, \mathbf{10})\}, \\
A_{1,2}^{0,1,2} &= \{(0, \mathbf{2}), (1, \mathbf{3}), (3, \mathbf{1}), (4, \mathbf{6})\}, \\
A_{2,2}^{0,1,2} &= \{(0, \mathbf{4}), (1, \mathbf{6}), (3, \mathbf{2}), (4, \mathbf{12})\}, \\
A_{1,3}^{0,1,2} &= \{(0, \mathbf{3}), (1, \mathbf{2}), (2, \mathbf{1}), (4, \mathbf{7})\}, \\
A_{2,3}^{0,1,2} &= \{(0, \mathbf{6}), (1, \mathbf{4}), (2, \mathbf{2}), (4, \mathbf{14})\}, \\
A_{1,4}^{0,1,2} &= \{(0, \mathbf{4}), (1, \mathbf{5}), (2, \mathbf{6}), (3, \mathbf{7})\}, \\
A_{2,4}^{0,1,2} &= \{(0, \mathbf{8}), (1, \mathbf{10}), (2, \mathbf{12}), (3, \mathbf{14})\}.
\end{align*}
\end{multicols}
\noindent Once all the pairs 
$A_{1,m}^{0,1,2},A_{2,m}^{0,1,2}$, for 
$m\in[n_v]$, are constructed, we analyze each pair separately. If 
$m=0$, notice that $(2,\mathbf{2})\in A_{1,0}^{0,1,2}$ and $(1,\mathbf{2})\in A_{2,0}^{0,1,2}$. Since there is a repetition in the second coordinate of these two elements and they differ in the first coordinate, we have $\rho_{A_{1,0}^{0,1,2},A_{2,0}^{0,1,2}}(\mathbf{2})=1$. Similarly, $(4,\mathbf{4})\in A_{1,0}^{0,1,2}$ and $(2,\mathbf{4})\in A_{2,0}^{0,1,2}$, so $\rho_{A_{1,0}^{0,1,2},A_{2,0}^{0,1,2}}(\mathbf{4})=1$.
After applying the same strategy to the other pairs, we obtain $\rho_{A_{1,1}^{0,1,2},A_{2,1}^{0,1,2}}(\mathbf{2})=1$, $\rho_{A_{1,2}^{0,1,2},A_{2,2}^{0,1,2}}(\mathbf{2})=1$, $\rho_{A_{1,2}^{0,1,2},A_{2,2}^{0,1,2}}(\mathbf{6})=1$, and $\rho_{A_{1,3}^{0,1,2},A_{2,3}^{0,1,2}}(\mathbf{2})=1$. Plugging in these values in \eqref{counting_cycles_formula_num6cycles_ncxnv_dyadics}, we have
\begin{equation*}
\mathcal{N}_6=2^4\cdot\left[(1+1)+(1)+(1+1)+(1)+(0)\right] = 16\cdot6 = 96,
\end{equation*}
so the number of 6-cycles in the Tanner graph is $\mathcal{N}_6=96$.
\end{example}

\subsubsection{Counting 8-cycles}
\label{counting_cycles_counting_8cycles_dyadics}

Let us start by simplifying the notation. Let $W$ be a walk in the protograph. 

\noindent If $W=[{\bh_{\balpha_1}}_{\beta_1},{\bh_{\balpha_2}}_{\beta_1},{\bh_{\balpha_3}}_{\beta_2},{\bh_{\balpha_4}}_{\beta_2},\dots,{\bh_{\balpha_{k-1}}}_{\beta_m},{\bh_{\balpha_{k}}}_{\beta_m}]$ has length $k=2m$, then we write
\begin{equation*}
[{\bh_{\balpha_1}},{\bh_{\balpha_2}},\dots,{\bh_{\balpha_{k}}}]_{\beta_1,\beta_2,\dots,\beta_m}.    
\end{equation*} 
If $W=[{\bh_{\balpha_1}}_{\beta_1},{\bh_{\balpha_2}}_{\beta_1},{\bh_{\balpha_3}}_{\beta_2},{\bh_{\balpha_4}}_{\beta_2},\dots,{\bh_{\balpha_{k-1}}}_{\beta_m},{\bh_{\balpha_{k}}}_{\beta_m},{\bh_{\balpha_{k+1}}}_{\beta_{m+1}}]$ has length $k=2m+1$, then we write 
\begin{equation*}
[{\bh_{\balpha_1}},{\bh_{\balpha_2}},\dots,{\bh_{\balpha_{k}}}|{\bh_{\balpha_{k+1}}}]_{\beta_1,\beta_2,\dots,\beta_m |\beta_{m+1}}.
\end{equation*}

The 8-cycles in the Tanner graph of protograph-based LDPC codes are projected onto a $4\times4$ submatrix of the form
\begin{equation*}
\label{parity_check_matrix_dyadics_4x4_submatrix}
H_4=\matrix{P_{\bh_{u}}&P_{\bh_{v}}&P_{\bh_{v'}}&P_{\bh_{u'}} \\ P_{\bi_{u}}&P_{\bi_{v}}&P_{\bi_{v'}}&P_{\bi_{u'}} \\ P_{\bj_{u}}&P_{\bj_{v}}&P_{\bj_{v'}}&P_{\bj_{u'}} \\ P_{\bk_{u}}&P_{\bk_{v}}&P_{\bk_{v'}}&P_{\bk_{u'}} }.
\end{equation*}
The corresponding 21 nonequivalent TBC walk patterns of length 8 in the protograph obtained from this submatrix are given by 
\begin{multicols}{3}
\noindent
\begin{align*}
&{[\bh,\bi,\bi,\bh,\bh,\bi,\bi,\bh]}_{u,v,v',u'}, \\
&{[\bh,\bi,\bi,\bh,\bh,\bj,\bj,\bh]}_{u,v,v',u'}, \\
&{[\bh,\bi,\bi,\bh,\bh,\bk,\bk,\bh]}_{u,v,v',u'}, \\
&{[\bh,\bj,\bj,\bh,\bh,\bj,\bj,\bh]}_{u,v,v',u'}, \\
&{[\bh,\bj,\bj,\bh,\bh,\bk,\bk,\bh]}_{u,v,v',u'}, \\
&{[\bh,\bk,\bk,\bh,\bh,\bk,\bk,\bh]}_{u,v,v',u'}, \\
&{[\bh,\bj,\bj,\bi,\bi,\bj,\bj,\bh]}_{u,v,v',u'}, \\
&{[\bh,\bj,\bj,\bi,\bi,\bk,\bk,\bh]}_{u,v,v',u'}, \\
&{[\bh,\bk,\bk,\bi,\bi,\bk,\bk,\bh]}_{u,v,v',u'}, \\
&{[\bh,\bi,\bi,\bj,\bj,\bi,\bi,\bh]}_{u,v,v',u'}, \\
&{[\bh,\bi,\bi,\bj,\bj,\bk,\bk,\bh]}_{u,v,v',u'}, \\
&{[\bh,\bk,\bk,\bj,\bj,\bk,\bk,\bh]}_{u,v,v',u'}, \\
&{[\bh,\bi,\bi,\bk,\bk,\bi,\bi,\bh]}_{u,v,v',u'}, \\
&{[\bh,\bi,\bi,\bk,\bk,\bj,\bj,\bh]}_{u,v,v',u'}, \\
&{[\bh,\bj,\bj,\bk,\bk,\bj,\bj,\bh]}_{u,v,v',u'}, \\
&{[\bi,\bj,\bj,\bi,\bi,\bj,\bj,\bi]}_{u,v,v',u'}, \\
&{[\bi,\bj,\bj,\bi,\bi,\bk,\bk,\bi]}_{u,v,v',u'}, \\
&{[\bi,\bk,\bk,\bi,\bi,\bk,\bk,\bi]}_{u,v,v',u'}, \\
&{[\bi,\bk,\bk,\bj,\bj,\bk,\bk,\bi]}_{u,v,v',u'}, \\
&{[\bi,\bj,\bj,\bk,\bk,\bj,\bj,\bi]}_{u,v,v',u'}, \\
&{[\bj,\bk,\bk,\bj,\bj,\bk,\bk,\bj]}_{u,v,v',u'},
\end{align*}
\end{multicols}
\noindent for $0\leq h<i<j<k\leq n_c-1$ and $u,v,v',u'\in[n_v]$. 
To count 8-cycles in the Tanner graph, we rewrite these conditions. Consider the following theorem.
\begin{theorem}
\label{counting_cycles_thm_8cycles_ncxnv_dyadics}
Let $H$ be as in \eqref{parity_check_matrix_dyadics}. For $0\leq h<i<j<k\leq n_c-1$, consider the following msets
\begin{align*}
A_{1,1}^{h,i,j,k} &= \left\{(u,v,\bh_u+\bi_u+\bi_v+\bh_v)\mid u,v\in[n_v], u\neq v\right\}, \\
A_{1,2}^{h,i,j,k} &= \left\{(u,v,\bh_u+\bi_u+\bi_v+\bh_v)\mid u,v\in[n_v], u\neq v\right\}, \\
A_{2,1}^{h,i,j,k} &= \left\{(u,v,\bh_u+\bi_u+\bi_v+\bh_v)\mid u,v\in[n_v], u\neq v\right\}, \\
A_{2,2}^{h,i,j,k} &= \left\{(u,v,\bh_u+\bj_u+\bj_v+\bh_v)\mid u,v\in[n_v], u\neq v\right\}, \\
A_{3,1}^{h,i,j,k} &= \left\{(u,v,\bh_u+\bi_u+\bi_v+\bh_v)\mid u,v\in[n_v], u\neq v\right\}, \\
A_{3,2}^{h,i,j,k} &= \left\{(u,v,\bh_u+\bk_u+\bk_v+\bh_v)\mid u,v\in[n_v], u\neq v\right\}, \\
A_{4,1}^{h,i,j,k} &= \left\{(u,v,\bh_u+\bj_u+\bj_v+\bh_v)\mid u,v\in[n_v], u\neq v\right\}, \\
A_{4,2}^{h,i,j,k} &= \left\{(u,v,\bh_u+\bj_u+\bj_v+\bh_v)\mid u,v\in[n_v], u\neq v\right\}, \\
A_{5,1}^{h,i,j,k} &= \left\{(u,v,\bh_u+\bj_u+\bj_v+\bh_v)\mid u,v\in[n_v], u\neq v\right\}, \\
A_{5,2}^{h,i,j,k} &= \left\{(u,v,\bh_u+\bk_u+\bk_v+\bh_v)\mid u,v\in[n_v], u\neq v\right\}, \\
A_{6,1}^{h,i,j,k} &= \left\{(u,v,\bh_u+\bk_u+\bk_v+\bh_v)\mid u,v\in[n_v], u\neq v\right\}, \\
A_{6,2}^{h,i,j,k} &= \left\{(u,v,\bh_u+\bk_u+\bk_v+\bh_v)\mid u,v\in[n_v], u\neq v\right\}, \\
A_{7,1}^{h,i,j,k} &= \left\{(u,v,\bh_u+\bj_u+\bj_v+\bi_v)\mid u,v\in[n_v], u\neq v\right\}, \\
A_{7,2}^{h,i,j,k} &= \left\{(u,v,\bh_u+\bj_u+\bj_v+\bi_v)\mid u,v\in[n_v], u\neq v\right\}, \\
A_{8,1}^{h,i,j,k} &= \left\{(u,v,\bh_u+\bj_u+\bj_v+\bi_v)\mid u,v\in[n_v], u\neq v\right\}, \\
A_{8,2}^{h,i,j,k} &= \left\{(u,v,\bh_u+\bk_u+\bk_v+\bi_v)\mid u,v\in[n_v], u\neq v\right\}, \\
A_{9,1}^{h,i,j,k} &= \left\{(u,v,\bh_u+\bk_u+\bk_v+\bi_v)\mid u,v\in[n_v], u\neq v\right\}, \\
A_{9,2}^{h,i,j,k} &= \left\{(u,v,\bh_u+\bk_u+\bk_v+\bi_v)\mid u,v\in[n_v], u\neq v\right\}, \\
A_{10,1}^{h,i,j,k} &= \left\{(u,v,\bh_u+\bi_u+\bi_v+\bj_v)\mid u,v\in[n_v], u\neq v\right\}, \\
A_{10,2}^{h,i,j,k} &= \left\{(u,v,\bh_u+\bi_u+\bi_v+\bj_v)\mid u,v\in[n_v], u\neq v\right\}, \\
A_{11,1}^{h,i,j,k} &= \left\{(u,v,\bh_u+\bi_u+\bi_v+\bj_v)\mid u,v\in[n_v], u\neq v\right\}, \\
A_{11,2}^{h,i,j,k} &= \left\{(u,v,\bh_u+\bk_u+\bk_v+\bj_v)\mid u,v\in[n_v], u\neq v\right\}, \\
A_{12,1}^{h,i,j,k} &= \left\{(u,v,\bh_u+\bk_u+\bk_v+\bj_v)\mid u,v\in[n_v], u\neq v\right\}, \\
A_{12,2}^{h,i,j,k} &= \left\{(u,v,\bh_u+\bk_u+\bk_v+\bj_v)\mid u,v\in[n_v], u\neq v\right\}, \\
A_{13,1}^{h,i,j,k} &= \left\{(u,v,\bh_u+\bi_u+\bi_v+\bk_v)\mid u,v\in[n_v], u\neq v\right\}, \\
A_{13,2}^{h,i,j,k} &= \left\{(u,v,\bh_u+\bi_u+\bi_v+\bk_v)\mid u,v\in[n_v], u\neq v\right\}, \\
A_{14,1}^{h,i,j,k} &= \left\{(u,v,\bh_u+\bi_u+\bi_v+\bk_v)\mid u,v\in[n_v], u\neq v\right\}, \\
A_{14,2}^{h,i,j,k} &= \left\{(u,v,\bh_u+\bj_u+\bj_v+\bk_v)\mid u,v\in[n_v], u\neq v\right\}, \\
A_{15,1}^{h,i,j,k} &= \left\{(u,v,\bh_u+\bj_u+\bj_v+\bk_v)\mid u,v\in[n_v], u\neq v\right\}, \\
A_{15,2}^{h,i,j,k} &= \left\{(u,v,\bh_u+\bj_u+\bj_v+\bk_v)\mid u,v\in[n_v], u\neq v\right\}, \\
A_{16,1}^{h,i,j,k} &= \left\{(u,v,\bi_u+\bj_u+\bj_v+\bi_v)\mid u,v\in[n_v], u\neq v\right\}, \\
A_{16,2}^{h,i,j,k} &= \left\{(u,v,\bi_u+\bj_u+\bj_v+\bi_v)\mid u,v\in[n_v], u\neq v\right\}, \\
A_{17,1}^{h,i,j,k} &= \left\{(u,v,\bi_u+\bj_u+\bj_v+\bi_v)\mid u,v\in[n_v], u\neq v\right\}, \\
A_{17,2}^{h,i,j,k} &= \left\{(u,v,\bi_u+\bk_u+\bk_v+\bi_v)\mid u,v\in[n_v], u\neq v\right\}, \\
A_{18,1}^{h,i,j,k} &= \left\{(u,v,\bi_u+\bk_u+\bk_v+\bi_v)\mid u,v\in[n_v], u\neq v\right\}, \\
A_{18,2}^{h,i,j,k} &= \left\{(u,v,\bi_u+\bk_u+\bk_v+\bi_v)\mid u,v\in[n_v], u\neq v\right\}, \\
A_{19,1}^{h,i,j,k} &= \left\{(u,v,\bi_u+\bk_u+\bk_v+\bj_v)\mid u,v\in[n_v], u\neq v\right\}, \\
A_{19,2}^{h,i,j,k} &= \left\{(u,v,\bi_u+\bk_u+\bk_v+\bj_v)\mid u,v\in[n_v], u\neq v\right\}, \\
A_{20,1}^{h,i,j,k} &= \left\{(u,v,\bi_u+\bj_u+\bj_v+\bk_v)\mid u,v\in[n_v], u\neq v\right\}, \\
A_{20,2}^{h,i,j,k} &= \left\{(u,v,\bi_u+\bj_u+\bj_v+\bk_v)\mid u,v\in[n_v], u\neq v\right\}, \\
A_{21,1}^{h,i,j,k} &= \left\{(u,v,\bj_u+\bk_u+\bk_v+\bj_v)\mid u,v\in[n_v], u\neq v\right\}, \\
A_{21,2}^{h,i,j,k} &= \left\{(u,v,\bj_u+\bk_u+\bk_v+\bj_v)\mid u,v\in[n_v], u\neq v\right\}. \\
\end{align*}
\noindent For $u,v,u',v'\in[n_v]$, let $(u,v,\balpha_{u,v})\in A_{1,1}^{h,i,j,k}$ and $(u',v',\balpha_{u',v'})\in A_{1,2}^{h,i,j,k}$ be such that $\balpha_{u,v}=\balpha_{u',v'}$. Then this repetition $\balpha_{u,v}=\balpha_{u',v'}$ lifts to a collection of 8-cycles in the Tanner graph if $u\neq u'$ and $v\neq v'$. The same result follows if the pair $A_{1,1}^{h,i,j,k}$, $A_{1,2}^{h,i,j,k}$ is replaced with any of the other twenty pairs. 
The total number of 8-cycles in the Tanner graph, $\mathcal{N}_8$, is given by 
\begin{align}
\begin{split}
\label{counting_cycles_formula_num8cycles_ncxnv_dyadics}
\mathcal{N}_8 = \; & \sum_{\substack{h,i,j,k\in[n_c]\\h<i<j<k}} \biggl( \mathcal{R}_{A_{1,1}^{h,i,j,k},A_{1,2}^{h,i,j,k}}^*+\frac{2^\ell}{2}\mathcal{R}_{A_{2,1}^{h,i,j,k},A_{2,2}^{h,i,j,k}}+\frac{2^\ell}{2}\mathcal{R}_{A_{3,1}^{h,i,j,k},A_{3,2}^{h,i,j,k}}+\mathcal{R}_{A_{4,1}^{h,i,j,k},A_{4,2}^{h,i,j,k}}^*, \\
& +\frac{2^\ell}{2}\mathcal{R}_{A_{5,1}^{h,i,j,k},A_{5,2}^{h,i,j,k}}+\mathcal{R}_{A_{6,1}^{h,i,j,k},A_{6,2}^{h,i,j,k}}^*+\frac{2^\ell}{2}\mathcal{R}_{A_{7,1}^{h,i,j,k},A_{7,2}^{h,i,j,k}}+2^\ell\mathcal{R}_{A_{8,1}^{h,i,j,k},A_{8,2}^{h,i,j,k}} \\
& +\frac{2^\ell}{2}\mathcal{R}_{A_{9,1}^{h,i,j,k},A_{9,2}^{h,i,j,k}}+\frac{2^\ell}{2}\mathcal{R}_{A_{10,1}^{h,i,j,k},A_{10,2}^{h,i,j,k}}+2^\ell\mathcal{R}_{A_{11,1}^{h,i,j,k},A_{11,2}^{h,i,j,k}}+\frac{2^\ell}{2}\mathcal{R}_{A_{12,1}^{h,i,j,k},A_{12,2}^{h,i,j,k}} \\
& +\frac{2^\ell}{2}\mathcal{R}_{A_{13,1}^{h,i,j,k},A_{13,2}^{h,i,j,k}}+2^\ell\mathcal{R}_{A_{14,1}^{h,i,j,k},A_{14,2}^{h,i,j,k}}+\frac{2^\ell}{2}\mathcal{R}_{A_{15,1}^{h,i,j,k},A_{15,2}^{h,i,j,k}}+\mathcal{R}_{A_{16,1}^{h,i,j,k},A_{16,2}^{h,i,j,k}}^* \\
& +\frac{2^\ell}{2}\mathcal{R}_{A_{17}^{h,i,j,k},A_{17}^{h,i,j,k}}+\mathcal{R}_{A_{18,1}^{h,i,j,k},A_{18,2}^{h,i,j,k}}^*+\frac{2^\ell}{2}\mathcal{R}_{A_{19,1}^{h,i,j,k},A_{19,2}^{h,i,j,k}}+\frac{2^\ell}{2}\mathcal{R}_{A_{20,1}^{h,i,j,k},A_{20,2}^{h,i,j,k}} \\
& +\mathcal{R}_{A_{21,1}^{h,i,j,k},A_{21,2}^{h,i,j,k}}^* \biggr),
\end{split}
\end{align}
where $\mathcal{R}_{X_1,X_2}$ is the number of repetitions $\balpha_{u,v}=\balpha_{u',v'}$ between the msets $X_1$ and $X_2$, the coefficient of each $\mathcal{R}_{X_1,X_2}$ is coming from the number of equivalent walks for the corresponding TBC walk pattern, $\mathcal{R}_{X_1,X_2}^*$ is given by
\begin{equation}
\label{counting_cycles_formula_num8cycles_ncxnv_dyadics_asterisk_pattern}
\mathcal{R}_{X_1,X_2}^* = \frac{2^\ell}{2}\cdot\frac{1}{2}\cdot\mathcal{R}_{X_1,X_2}^{*c}+2^\ell\cdot\frac{1}{4}\mathcal{R}_{X_1,X_2}^{*nc} = \frac{2^\ell}{4}\cdot\left(\mathcal{R}_{X_1,X_2}^{*c}+\mathcal{R}_{X_1,X_2}^{*nc}\right),
\end{equation}
and $\mathcal{R}_{X_1,X_2}^{*c}$ and $\mathcal{R}_{X_1,X_2}^{*nc}$ are the numbers of repetitions $\balpha_{u,v}=\balpha_{u',v'}$ between the msets $X_1$ and $X_2$ satisfying and not satisfying the conditions $u=v'$ and $v=u'$, respectively.
\end{theorem}
\begin{proof}
By Lemma \ref{counting_cycles_dyadics_imagecycle}, the 8-cycles in the Tanner graph are projected onto either TBC walks of length 4 that are traversed twice or TBC walks of length 8, in the protograph, and these walks can be studied by computing the product $HH^\mathsf{T}HH^\mathsf{T}$ as a consequence of Theorem \ref{counting_cycles_dyadics_polynomialpowerofadjacencymatrix}. To determine all possible TBC walk patterns, we use the product $H_4H_4^\mathsf{T}H_4H_4^\mathsf{T}$ instead, but if $H$ has less than four rows, then any pattern arising from these additional rows are omitted in the computations of 8-cycles. Once $H_4H_4^\mathsf{T}H_4H_4^\mathsf{T}$ has been computed, we use the strategy of combining walks following Definition \ref{counting_cycles_dyadics_permutationshift} to compute all possible TBC walk patterns, and then use Definition \ref{equivalent_closed_walks} to compute all nonequivalent TBC walk patterns. This is how we obtained the 21 nonequivalent TBC walk patterns before the statement of Theorem \ref{counting_cycles_thm_8cycles_ncxnv_dyadics} and the corresponding 21 pairs of sets $A_{m,1}^{h,i,j,k},A_{m,2}^{h,i,j,k}$ for $1\leq m\leq21$.

To show that a repetition in any of these pairs lifts to a collection of 8-cycles in the Tanner graph, consider a repetition in the pair $A_{1,1}^{h,i,j,k},A_{1,2}^{h,i,j,k}$. There are $u,v,u',v'\in[n_v]$, $(u,v,\balpha_{u,v})\in A_{1,1}^{h,i,j,k}$, and $(u',v',\balpha_{u',v'})\in A_{1,2}^{h,i,j,k}$ such that $\balpha_{u,v}=\balpha_{u',v'}$. Then this is equivalent to $\bh_u+\bi_u+\bi_v+\bh_v+\bh_{v'}+\bi_{v'}+\bi_{u'}+\bh_{u'}=\bzero$. To guarantee that this permutation shift represents a TBC walk, in addition to the conditions $u\neq v$ and $u'\neq v'$ required in the construction of the msets $A_{1,1}^{h,i,j,k}$ and $A_{1,2}^{h,i,j,k}$, we need to ensure that $u\neq u'$ and $v\neq v'$. In this case, this TBC walk lifts to a collection of 8-cycles. The same approach works for the remaining 20 pairs. 

Once a TBC walk pattern is fixed, it is possible that at least one of its  equivalent walks has the same pattern, which means that the same collection of 8-cycles will be counted multiple times. Let us focus on the indices $(u,v,v',u')$. If we let $m=1$, then for the corresponding TBC walk pattern, there are four equivalent walks that have the same TBC walk pattern and these have index tuples $(u,v,v',u')$, $(v',u',u,v)$, $(v,u,u',v')$, and $(u',v',v,u)$. From this collection of index tuples, we get the relation $u=v'$ and $v=u'$. If this relation is satisfied, then we obtain two distinct index tuples $(u,v,u,v)$ and $(v,u,v,u)$, so the corresponding coefficient is $\frac{1}{2}$ and their contribution to the number of 8-cycles is $\frac{2^\ell}{2}$ (since they correspond to the double traversal of a TBC walk of length 4). If the relation is not satisfied, the four index tuples are distinct, so the corresponding coefficient is $\frac{1}{4}$ and their contribution to the number of 8-cycles is $2^\ell$. The same situation happens for $m\in\{4,6,16,18,21\}$.

If we let $m=2$, then for the corresponding TBC walk pattern, there are two equivalent walks that have the same TBC walk pattern and these have index tuples $(u,v,v',u')$ and $(v,u,u',v')$. Since there is no possible relation between the index tuples, they are always distinct, so the corresponding coefficient is $\frac{1}{2}$ and their contribution to the number of 8-cycles is $2^\ell$. The same situation happens for $m\in\{3,5,7,9,10,11,12,13,14,15,17,19,20\}$. 

Finally, if we let $m=8$, then for the corresponding TBC walk pattern, there are no equivalent walks that have the same TBC walk pattern, so this gives the unique index tuple $(u,v,v',u')$. Naturally, the corresponding coefficient is $1$ and their contribution to the number of 8-cycles is $2^\ell$. The same situation happens for $m\in\{11,14\}$.

Therefore, if we let $\mathcal{R}_{X_1,X_2}$ be the number of repetitions $\balpha_{u,v}=\balpha_{u',v'}$ between the msets $X_1$ and $X_2$, and $\mathcal{R}_{X_1,X_2}^{*c}$ and $\mathcal{R}_{X_1,X_2}^{*nc}$ be the numbers of repetitions $\balpha_{u,v}=\balpha_{u',v'}$ between the msets $X_1$ and $X_2$ satisfying and not satisfying the conditions $u=v'$ and $v=u'$, respectively, then \eqref{counting_cycles_formula_num8cycles_ncxnv_dyadics} and \eqref{counting_cycles_formula_num8cycles_ncxnv_dyadics_asterisk_pattern} follow, and we conclude the proof.
\end{proof}

\begin{remark}
As mentioned in the proof of Theorem \ref{counting_cycles_thm_8cycles_ncxnv_dyadics}, the product $H_4H_4^\mathsf{T}H_4H_4^\mathsf{T}$ was used to determine all possible TBC walk patterns that can contribute to 8-cycles in the Tanner graph instead of the product $HH^\mathsf{T}HH^\mathsf{T}$. If $H$ has less than four rows, then any pattern arising from these additional rows are omitted in the computations of 8-cycles. If $H$ has one row, notice that all TBC walk patterns are omitted, so there cannot be any 8-cycles. If $H$ has two rows numbered by $h$ and $i$, then the only TBC walk patterns involved in the computation of 8-cycles have $m\in\{1\}$. Similarly, if $H$ has three rows numbered by $h$, $i$, and $j$, then the only TBC walk patterns involved in the computation of 8-cycles have $m\in\{1,2,4,7,10,16\}$.
\end{remark}

\begin{example}
\label{counting_cycles_example_8cycles_g6_dyadics}
Let $H$ be the parity-check matrix given in Example  
\ref{counting_cycles_example_6cycles_g6_dyadics}. Since this parity-check matrix has girth 6 for $N=2^4$, we can use our strategy to count the number of 8-cycles in the Tanner graph. We construct the msets in Theorem \ref{counting_cycles_thm_8cycles_ncxnv_dyadics} and check for repetitions. Since $H$ has exactly three rows, the only combination of three rows is given by $(h,i,j)=(0,1,2)$, 
so we only need to construct the six pairs $A_{m,1}^{0,1,2},A_{m,1}^{0,1,2}$ for $m\in\{1,2,4,7,10,16\}$. These computations have been omitted for space constraints. 
\noindent After constructing these msets, we analyze each pair separately and look for repetitions. Some computations show that 
we have
\begin{align*}
\begin{split}
\mathcal{N}_8 = \; & \mathcal{R}_{A_{1,1},A_{1,2}}^*+\frac{2^\ell}{2}\mathcal{R}_{A_{2,1},A_{2,2}}+\mathcal{R}_{A_{4,1},A_{4,2}}^* +\frac{2^\ell}{2}\mathcal{R}_{A_{7,1},A_{7,2}}+\frac{2^\ell}{2}\mathcal{R}_{A_{10,1},A_{10,2}}+\mathcal{R}_{A_{16,1},A_{16,2}}^*
\end{split} \\
\begin{split}
= \; & \frac{2^\ell}{4}\left(\mathcal{R}_{A_{1,1},A_{1,2}}^{*c}+\mathcal{R}_{A_{1,1},A_{1,2}}^{*nc}\right)+\frac{2^\ell}{2}\mathcal{R}_{A_{2,1},A_{2,2}}+\frac{2^\ell}{4}\left(\mathcal{R}_{A_{4,1},A_{4,2}}^{*c}+\mathcal{R}_{A_{4,1},A_{4,2}}^{*nc}\right) \\ &+\frac{2^\ell}{2}\mathcal{R}_{A_{7,1},A_{7,2}}+\frac{2^\ell}{2}\mathcal{R}_{A_{10,1},A_{10,2}}+\frac{2^\ell}{4}\left(\mathcal{R}_{A_{16,1},A_{16,2}}^{*c}+\mathcal{R}_{A_{16,1},A_{16,2}}^{*nc}\right)
\end{split} \\
= \; & \frac{2^{4}}{4}\left(20+24\right)+\frac{2^{4}}{2}(20)+\frac{2^{4}}{4}\left(20+24\right) +\frac{2^{4}}{2}(12)+\frac{2^{4}}{2}(20)+\frac{2^{4}}{4}\left(20+24\right) \\
= \; & 944,
\end{align*}
so the number of 8-cycles in the Tanner graph is $\mathcal{N}_8=944$.
\end{example}


\bibliographystyle{IEEEtran}
\bibliography{dyadic}

\end{document}